\documentclass{bmcart}

\usepackage{amsthm,amsmath}
\usepackage[utf8]{inputenc} 

\usepackage{amssymb}
\usepackage{graphicx}
\usepackage{subfigure}
\usepackage[linesnumbered,ruled,vlined]{algorithm2e}
\usepackage{color,soul}
\usepackage{hyperref}

\startlocaldefs
\endlocaldefs

\newtheorem{definition}{Definition}
\newtheorem{remark}{Remark}
\newtheorem{theorem}{Theorem}
\newtheorem{corollary}{Corollary}
\newtheorem{assumption}{Assumption}

\SetKwInput{KwInput}{Input}                
\SetKwInput{KwOutput}{Output}

\begin{document}

\begin{frontmatter}

\begin{fmbox}
\dochead{Research}


\title{Dynamic Network Sampling for Community Detection}


\author[
addressref={aff1}, 
corref={aff1}, 
email={cmu2@jhu.edu}
]{\inits{C.}\fnm{Cong} \snm{Mu}}
\author[
addressref={aff2}, 
email={youngser@jhu.edu}
]{\inits{Y.}\fnm{Youngser} \snm{Park}}
\author[
addressref={aff1}, 
email={cep@jhu.edu}
]{\inits{C.E.}\fnm{Carey E.} \snm{Priebe}}


\address[id=aff1]{
	\orgdiv{Department of Applied Mathematics and Statistics},             
	\orgname{Johns Hopkins University},          
	\city{Baltimore},                              
	\cny{US}                                    
}
\address[id=aff2]{%
	\orgdiv{Center for Imaging Science},
	\orgname{Johns Hopkins University},
	\city{Baltimore},
	\cny{US}
}



\end{fmbox}


\begin{abstractbox}

\begin{abstract} 
We propose a dynamic network sampling scheme to optimize block recovery for stochastic blockmodel (SBM) in the case where it is prohibitively expensive to observe the entire graph. Theoretically, we provide justification of our proposed Chernoff-optimal dynamic sampling scheme via the Chernoff information. Practically, we evaluate the performance, in terms of block recovery, of our method on several real datasets from different domains. Both theoretically and practically results suggest that our method can identify vertices that have the most impact on block structure so that one can only check whether there are edges between them to save significant resources but still recover the block structure.
\end{abstract}


\begin{keyword}
\kwd{dynamic network sampling}
\kwd{stochastic blockmodel}
\kwd{community detection}
\kwd{Chernoff information}
\end{keyword}


\end{abstractbox}
%

\end{frontmatter}



\section{Introduction}
In network inference applications, it is important to detect community structure, i.e., cluster vertices into potential blocks. However, it can be prohibitively expensive to observe the entire graph in many cases, especially for large graphs. For example, in a network where vertices represent landline phones and edges represent whether there is a call between two landline phones. Based on the size of the network, in terms of the number of vertices, it can be extremely expensive to check whether there is a call for every landline phone pairs. Therefore, if one can utilize the information carried by a partially oberverd graph, that is only a small number of landline phone pairs are verified, to identify the landline phones that may play a more important role in formulating communities. Then given limited resources, one can choose to only check whether there are calls between those landline phone pairs to achieve the goal of detecting potential block structure. Thus it becomes essential to identify vertices that have the most impact on block structure and only check whether there are edges between them to save significant resources but still recover the block structure.

Many classical methods only consider the adjacency or Laplacian matrices for community detection~\cite{Fortunato2016}. By contrast, vertex covariates can also be taken into consideration for the inference. These covariate-aware methods rely on either variational methods~\cite{Choi2012,Roy2019,Sweet2015} or spectral approaches~\cite{Binkiewicz2017,Huang2018,Mele2019,Mu2022}. However, none of them focus on the problem of clustering vertices for partially observed graphs. To address this issue, existing methods propose different types of random and adaptive sampling strategies to minimize the information loss from the data reduction~\cite{Yun2014,Purohit2017}.

We propose a dynamic network sampling scheme to optimize block recovery for stochastic blockmodel (SBM) when we only have limited resources to check whether there are edges between certain selected vertices. The innovation of our approach is the application of Chernoff information. To our knowledge, this is the first time that it has been applied to network sampling problems. Motivated by the Chernoff analysis, we not only propose a dynamic network sampling scheme to optimize block recovery, but also provide the framework and justification for using Chernoff information in subsequent inference for graphs.

The structure of this article is summarized as follows. Section~\ref{sec:2} reviews relevant models for random graphs and the basic idea of spectral methods. Section~\ref{sec:3} introduces the notion of Chernoff analysis for analytically measuring the performance of block recovery. Section~\ref{sec:4} includes our dynamic network sampling scheme and theoretical results. Section~\ref{sec:5} provides simulations and real data experiments to measure the algorithms' performance in terms of actual block recovery results. Section~\ref{sec:6} discusses the findings and presents some open questions for further investigation. Appendix provides technical details for our theoretical results.

\section{Models and Spectral Methods}

\label{sec:2}

In this work, we are interested in the inference task of block recovery (community detection). To model the block structure in edge-independent random graphs, we focus on the SBM and the generalized random dot product graph (GRDPG).

\begin{definition}[Generalized Random Dot Product Graph~\cite{Rubin-Delanchy2017}]
	\label{def:GRDPG}
	Let $ \mathbf{I}_{d_+ d_-} = \mathbf{I}_{d_+} \bigoplus \left(-\mathbf{I}_{d_-} \right)  $ with $ d_+ \geq 1 $ and $ d_- \geq 0 $. Let $ F $ be a $ d $-dimensional inner product distirbution with $ d = d_+ + d_- $ on $ \mathcal{X} \subset \mathbb{R}^d $ satisfying $ \mathbf{x}^\top  \mathbf{I}_{d_+ d_-} \mathbf{y} \in [0, 1] $ for all $ \mathbf{x}, \mathbf{y} \in \mathcal{X} $. Let $ \mathbf{A} $ be an adjacency matrix and $ \mathbf{X} = [\mathbf{X}_1, \cdots, \mathbf{X}_n]^\top \in \mathbb{R}^{n \times d} $ where $ \mathbf{X}_i \sim F $, i.i.d. for all $ i \in \{ 1, \cdots, n \} $. Then we say $ (\mathbf{A}, \mathbf{X}) \sim \text{GRDPG}(n, F, d_+, d_-) $ if for any $ i, j \in \{ 1, \cdots, n \} $
	\begin{equation}
	\mathbf{A}_{ij} \sim \text{Bernoulli}(\mathbf{P}_{ij}) \qquad \text{where} \qquad \mathbf{P}_{ij} = \mathbf{X}_{i}^\top \mathbf{I}_{d_+ d_-} \mathbf{X}_j.
	\end{equation}
\end{definition}

\begin{definition}[$ K $-block Stochastic Blockmodel Graph~\cite{Holland1983}]
	\label{def:SBM}
	The $ K $-block stochastic blockmodel (SBM) graph is an edge-independent random graph with each vertex belonging to one of $ K $ blocks. It can be parametrized by a block connectivity probability matrix $ \mathbf{B} \in (0, 1)^{K \times K} $ and a vector of block assignment probabilities $ \boldsymbol{\pi} \in (0, 1)^K $ summing to unity. Let $ \mathbf{A} $ be an adjacency matrix and $ \boldsymbol{\tau} $ be a vector of block assignments with $ \tau_i = k $ if vertex $ i $ is in block $ k $ (occuring with probability $ \pi_k $). We say $ (\mathbf{A}, \boldsymbol{\tau}) \sim \text{SBM}(n, \mathbf{B}, \boldsymbol{\pi}) $ if for any $ i, j \in \{ 1, \cdots, n \} $
	\begin{equation}
	\mathbf{A}_{ij} \sim \text{Bernoulli}(\mathbf{P}_{ij}) \qquad \text{where} \qquad \mathbf{P}_{ij} = \mathbf{B}_{\tau_i \tau_j}.
	\end{equation}
\end{definition}

\begin{remark}
	\label{remark:GRDPG-SBM}
	The SBM is a special case of the GRDPG model. Let $ (\mathbf{A}, \boldsymbol{\tau}) \sim \text{SBM}(n, \mathbf{B}, \boldsymbol{\pi}) $ as in Definition~\ref{def:SBM} where $ \mathbf{B} \in (0, 1)^{K \times K} $ with $ d_+ $ strictly positive eigenvalues and $ d_- $ strictly negative eigenvalues. To represent this SBM in the GRDPG model, we can choose $ \boldsymbol{\nu}_1, \cdots, \boldsymbol{\nu}_K \in \mathbb{R}^d $ where $ d = d_+ + d_- $ such that $ \boldsymbol{\nu}_k^\top \mathbf{I}_{d_+ d_-} \boldsymbol{\nu}_\ell = \mathbf{B}_{k \ell} $ for all $ k, \ell \in \{ 1, \cdots, K \} $. For example, we can take $ \boldsymbol{\nu} = \mathbf{U}_B |\mathbf{S}_B|^{1/2} $ where $ \mathbf{B} = \mathbf{U}_B \mathbf{S}_B \mathbf{U}_B^\top $ is the spectral decomposition of $ \mathbf{B} $ after re-ordering. Then we have the latent position of vertex $ i $ as $ \mathbf{X}_i = \boldsymbol{\nu}_k $ if $ \tau_i = k $.
\end{remark}

The parameters of the models can be estimated via spectral methods~\cite{Von2007}, which have been widely used in random graph models for community detection~\cite{Lyzinski2014,Lyzinski2016,McSherry2001,Rohe2011}. Two particular spectral embedding methods, adjacency spectral embedding (ASE) and Laplacian spectral embedding (LSE), are popular since they enjoy nice propertices including consistency~\cite{Sussman2012} and asymptotic normality~\cite{Athreya2016,Tang2018}.

\begin{definition}[Adjacency Spectral Embedding]
	Let $ \mathbf{A} \in \{0, 1 \}^{n \times n} $ be an adjacency matrix with eigendecomposition $ \mathbf{A} = \sum_{i=1}^{n} \lambda_i \mathbf{u}_i \mathbf{u}_i^\top $ where $ |\lambda_1| \geq \cdots \geq |\lambda_n| $ are the magnitude-ordered eigenvalues and $ \mathbf{u}_1, \cdots, \mathbf{u}_n $ are the corresponding orthonormal eigenvectors. Given the embedding dimension $ d < n $, the adjacency spectral embedding (ASE) of $ \mathbf{A} $ into $ \mathbb{R}^d $ is the $ n \times d $ matrix $ \mathbf{\widehat{X}} = \mathbf{U}_A |\mathbf{S}_A|^{1/2} $ where $ \mathbf{S}_A = \text{diag}(\lambda_1, \cdots, \lambda_d) $ and $ \mathbf{U}_A = [\mathbf{u}_1 | \cdots | \mathbf{u}_d] $.
\end{definition}

\begin{remark}
	\label{remark:dhat}
	There are different methods for choosing the embedding dimension~\cite{Hastie2009,Jolliffe2016}; we adopt the simple and efficient profile likelihood method~\cite{Zhu2006} to automatically identify ``elbow", which is the cut-off between the signal dimensions and the noise dimensions in scree plot.
\end{remark}

\section{Chernoff Analysis}

\label{sec:3}

To analytically measure the performance of algorithms for block recovery, we consider the notion of Chernoff information among other possible metrics. Chernoff information enjoys the advantages of being independent of the clustering procedure, i.e., it can be derived no matter which clustering methods are used, and it is intrinsically relating to the Bayes risk~\cite{Tang2018,Athreya2017,Karrer2011}.

\begin{definition}[Chernoff Information~\cite{Chernoff1952,Chernoff1956}]
	Let $ F_1 $ and $ F_2 $ be two continuous multivariate distributions on $ \mathbb{R}^d $ with density functions $ f_1 $ and $ f_2 $. The Chernoff information is defined as
	\begin{equation}
	\begin{split}
	C(F_1, F_2) & = - \log \left[\inf_{t \in (0,1)} \int_{\mathbb{R}^d} f_1^t(\mathbf{x}) f_2^{1-t}(\mathbf{x}) d\mathbf{x} \right] \\
	& = \sup_{t \in (0, 1)} \left[- \log \int_{\mathbb{R}^d} f_1^t(\mathbf{x}) f_2^{1-t}(\mathbf{x}) d\mathbf{x} \right].
	\end{split}
	\end{equation}
\end{definition}

\begin{remark}
	Consider the special case where we take $ F_1 = \mathcal{N}(\boldsymbol{\mu}_1, \boldsymbol{\Sigma}_1) $ and $ F_2 = \mathcal{N}(\boldsymbol{\mu}_2, \boldsymbol{\Sigma}_2) $; then the corresponding Chernoff information is
	\begin{equation}
	C(F_1, F_2) = \sup_{t \in (0, 1)} \left[ \frac{1}{2} t (1-t) (\boldsymbol{\mu}_1 - \boldsymbol{\mu}_2)^\top \boldsymbol{\Sigma}_t^{-1} (\boldsymbol{\mu}_1 - \boldsymbol{\mu}_2) + \frac{1}{2} \log \frac{\lvert \boldsymbol{\Sigma}_t \rvert}{\lvert \boldsymbol{\Sigma}_1 \rvert^t \lvert \boldsymbol{\Sigma}_2 \rvert^{1-t}} \right],
	\end{equation}
	where $ \boldsymbol{\Sigma}_t = t \boldsymbol{\Sigma}_1 + (1-t) \boldsymbol{\Sigma}_2 $.
\end{remark}

The comparsion of block recovery via Chernoff information is based on the statistical information between the limiting distributions of the blocks and smaller statistical information implies less information to discriminate between different blocks of the SBM. To that end, we also review the limiting results of ASE for SBM, essential for investigating Chernoff information.

\begin{theorem}[CLT of ASE for SBM~\cite{Rubin-Delanchy2017}]
	\label{thm:CLT-ASE-SBM}
	Let $ (\mathbf{A}^{(n)}, \mathbf{X}^{(n)}) \sim \text{GRDPG}(n, F, d_+, d_-) $ be a sequence of adjacency matrices and associated latent positions of a $ d $-dimensional GRDPG as in Definition~\ref{def:GRDPG} from an inner product distribution $ F $ where $ F $ is a mixture of $ K $ point masses in $ \mathbb{R}^d $, i.e.,
	\begin{equation}
	F = \sum_{k=1}^{K} \pi_k \delta_{\boldsymbol{\nu}_k} \qquad \text{with} \qquad \forall k, \; \pi_k > 0 \quad \text{and} \quad \sum_{k=1}^{K} \pi_k = 1,
	\end{equation}
	where $ \delta_{\boldsymbol{\nu}_k} $ is the Dirac delta measure at $ \nu_k $. Let $ \Phi(\mathbf{z}, \boldsymbol{\Sigma}) $ denote the cumulative distribution function (CDF) of a multivariate Gaussian distribution with mean $ \boldsymbol{0} $ and covariance matrix $ \boldsymbol{\Sigma} $, evaluated at $ \mathbf{z} \in \mathbb{R}^d $. Let $ \mathbf{\widehat{X}}^{(n)} $ be the ASE of $ \mathbf{A}^{(n)} $ with $ \mathbf{\widehat{X}}^{(n)}_i $ as the $ i $-th row (same for $ \mathbf{X}^{(n)}_i $). Then there exists a sequence of matrices $ \mathbf{M}_n \in \mathbb{R}^{d \times d} $ satisfying $ \mathbf{M}_n \mathbf{I}_{d_+ d_-} \mathbf{M}_n^\top = \mathbf{I}_{d_+ d_-} $ such that for all $ \mathbf{z} \in \mathbb{R}^d $ and fixed index i,
	\begin{equation}
	\mathbb{P} \left\{ \sqrt{n} \left(\mathbf{M}_n \mathbf{\widehat{X}}^{(n)}_i - \mathbf{X}^{(n)}_i \right) \leq \mathbf{z} \; \big| \; \mathbf{X}^{(n)}_i = \boldsymbol{\nu}_k  \right\} \to \Phi(\mathbf{z}, \boldsymbol{\Sigma}_k),
	\end{equation}
	where for $ \boldsymbol{\nu} \sim F $
	\begin{equation}
	\label{eq:Sigmax}
	\boldsymbol{\Sigma}_k = \boldsymbol{\Sigma}(\boldsymbol{\nu}_k) = \mathbf{I}_{d_+ d_-} \boldsymbol{\Delta}^{-1} \mathbb{E} \left[ \left(\boldsymbol{\nu}_k^\top \mathbf{I}_{d_+ d_-} \boldsymbol{\nu} \right) \left(1-\boldsymbol{\nu}_k^\top \mathbf{I}_{d_+ d_-} \boldsymbol{\nu} \right) \boldsymbol{\nu} \boldsymbol{\nu}^\top \right] \boldsymbol{\Delta}^{-1} \mathbf{I}_{d_+ d_-},
	\end{equation}
	with
	\begin{equation}
	\label{eq:Delta}
	\boldsymbol{\Delta} = \mathbb{E} \left[ \boldsymbol{\nu} \boldsymbol{\nu}^\top \right].
	\end{equation}
\end{theorem}

For a $ K $-block SBM, let $ \mathbf{B} \in (0, 1)^{K \times K} $ be the block connectivity probability matrix and $ \boldsymbol{\pi} \in (0, 1)^K $ be the vector of block assignment probabilities. Given an $ n $ vertex instantiation of the SBM parameterized by $ \mathbf{B} $ and $ \boldsymbol{\pi} $, for sufficiently large $ n $, the large sample optimal error rate for estimating the block assignments using ASE can be measured via Chernoff information as~\cite{Tang2018,Athreya2017}
\begin{equation}
\label{eq:rho}
\rho = \min_{k \neq l} \sup_{t \in (0, 1)} \left[ \frac{1}{2} n t (1-t) (\boldsymbol{\nu}_k - \boldsymbol{\nu}_\ell)^\top \boldsymbol{\Sigma}_{k\ell}^{-1}(t) (\boldsymbol{\nu}_k - \boldsymbol{\nu}_\ell) + \frac{1}{2} \log \frac{\lvert \boldsymbol{\Sigma}_{k \ell}(t) \rvert}{\lvert \boldsymbol{\Sigma}_k \rvert^t \lvert \boldsymbol{\Sigma}_\ell \rvert^{1-t}} \right],
\end{equation}
where $ \boldsymbol{\Sigma}_{k\ell}(t) = t \boldsymbol{\Sigma}_k + (1-t) \boldsymbol{\Sigma}_\ell $, $ \boldsymbol{\Sigma}_k = \boldsymbol{\Sigma}(\boldsymbol{\nu}_k) $ and $ \boldsymbol{\Sigma}_\ell = \boldsymbol{\Sigma}(\boldsymbol{\nu}_\ell) $ are defined as in Eq.~\eqref{eq:Sigmax}. Also note that as $ n \to \infty $, the logarithm term in Eq.~\eqref{eq:rho} will be dominated by the other term. Then we have the approximate Chernoff information as
\begin{equation}
\label{eq:rhoapprox}
\rho \approx \min_{k \neq l} C_{k ,\ell}(\mathbf{B}, \boldsymbol{\pi}),
\end{equation}
where
\begin{equation}
\label{eq:C_kl}
C_{k ,\ell}(\mathbf{B}, \boldsymbol{\pi}) =\sup_{t \in (0, 1)} \left[ t (1-t) (\boldsymbol{\nu}_k - \boldsymbol{\nu}_\ell)^\top \boldsymbol{\Sigma}_{k\ell}^{-1}(t) (\boldsymbol{\nu}_k - \boldsymbol{\nu}_\ell) \right].
\end{equation}

We also introduce the following two notions, which will be used when we describe our dynamic network sampling scheme.

\begin{definition}[Chernoff-active Blocks]
	For $K$-block SBM parametrized by the block connectivity probability matrix $ \mathbf{B} \in (0, 1)^{K \times K} $ and the vector of block assignment probabilities $ \boldsymbol{\pi} \in (0, 1)^K $. The Chernoff-active blocks $ (k^*, \ell^*) $ are defined as 
	\begin{equation}
	(k^*, \ell^*) = \arg \min_{k \neq l} C_{k ,\ell}(\mathbf{B}, \boldsymbol{\pi}),
	\end{equation}
	where $ C_{k ,\ell}(\mathbf{B}, \boldsymbol{\pi}) $ is defined as in Eq.~\eqref{eq:rhoapprox}.
\end{definition}

\begin{definition}[Chernoff Superiority]
	For $K$-block SBMs, given two block connectivity probability matrices $ \mathbf{B}, \mathbf{B}^\prime \in (0, 1)^{K \times K} $ and a vector of block assignment probabilities $ \boldsymbol{\pi} \in (0, 1)^K $. Let $ \rho_B $ and $ \rho_{B^\prime} $ denote the Chernoff information obtained as in Eq.~\eqref{eq:rhoapprox} corresponding to $ \mathbf{B} $ and $ \mathbf{B}^\prime $ respectively. We say that $ \mathbf{B} $ is Chernoff superior to $ \mathbf{B}^\prime $, denoted as $ \mathbf{B} \succ \mathbf{B}^\prime $, if $ \rho_B > \rho_{B^\prime} $.
\end{definition}

\begin{remark}
	If $ \mathbf{B} $ is Chernoff superior to $ \mathbf{B}^\prime $, then we can have a better block recovery from $ \mathbf{B} $ than $ \mathbf{B}^\prime $. In addition, Chernoff superiority is transitive, which is straightforward from the definition.
\end{remark}

\section{Dynamic Network Sampling}

\label{sec:4}

We start our analysis with the unobserved block connectivity probability matrix $ \mathbf{B} $ for SBM and then illustrate how to migrate the proposed methods for real applications when we have the observed adjacency matrix $ \mathbf{A} $.

Consider the $ K $-block SBM parametrized by the block connectivity probability matrix $ \mathbf{B} \in (0, 1)^{K \times K} $ and the vector of block assignment probabilities $ \boldsymbol{\pi} \in (0, 1)^K $ with $ K > 2 $. Given initial sampling parameter $ p_0 \in (0, 1) $, initial sampling is uniformly at random, i.e., 
\begin{equation}
\label{eq:B0}
\mathbf{B}_0 = p_0 \mathbf{B}.
\end{equation}

This initial sampling simulates the case when one only obersves a partial graph with a small portion of the edges instead of the entire graph with all existing edges.

\begin{theorem}
	\label{thm:Chernoff-Superiority}
	For $K$-block SBMs, given two block connectivity probability matrices $ \mathbf{B}, p\mathbf{B} \in (0, 1)^{K \times K} $ with $ p \in (0, 1) $ and a vector of block assignment probabilities $ \boldsymbol{\pi} \in (0, 1)^K $, we have $ \mathbf{B} \succ p \mathbf{B} $.
\end{theorem}

The proof of Theorem~\ref{thm:Chernoff-Superiority} can be found in Appendix. As an illustration, consider a 4-block SBM parametrized by block connectivity probability matrix $ \mathbf{B} $ as
\begin{equation}
\label{eq:exampleB}
\mathbf{B} = 
\begin{bmatrix}
0.04 & 0.08 & 0.10 & 0.18 \\
0.08 & 0.16 & 0.20 & 0.36 \\
0.10 & 0.20 & 0.25 & 0.45 \\
0.18 & 0.36 & 0.45 & 0.81
\end{bmatrix}.
\end{equation}	

Figure~\ref{fig:rho0} shows Chernoff information $ \rho $ as in Eq.~\eqref{eq:rhoapprox} corresponding to $ \mathbf{B} $ as in Eq.~\eqref{eq:exampleB} and $ p \mathbf{B} $ for $ p \in  (0, 1) $. In addition, Figure~\ref{fig:rho0a} assumes $ \boldsymbol{\pi} = (\frac{1}{4}, \frac{1}{4}, \frac{1}{4}, \frac{1}{4}) $ and Figure~\ref{fig:rho0b} assumes $ \boldsymbol{\pi} = (\frac{1}{8}, \frac{1}{8}, \frac{3}{8}, \frac{3}{8}) $. As suggested by Theorem~\ref{thm:Chernoff-Superiority}, for any $ p \in (0, 1) $ we have $\rho_{B} > \rho_{pB} $ and thus $ \mathbf{B} \succ p \mathbf{B} $.

\begin{figure}[h!]
	\subfigure[balanced: $ \boldsymbol{\pi} = (\frac{1}{4}, \frac{1}{4}, \frac{1}{4}, \frac{1}{4}) $ \label{fig:rho0a}]{
		\includegraphics[width=0.45\textwidth]{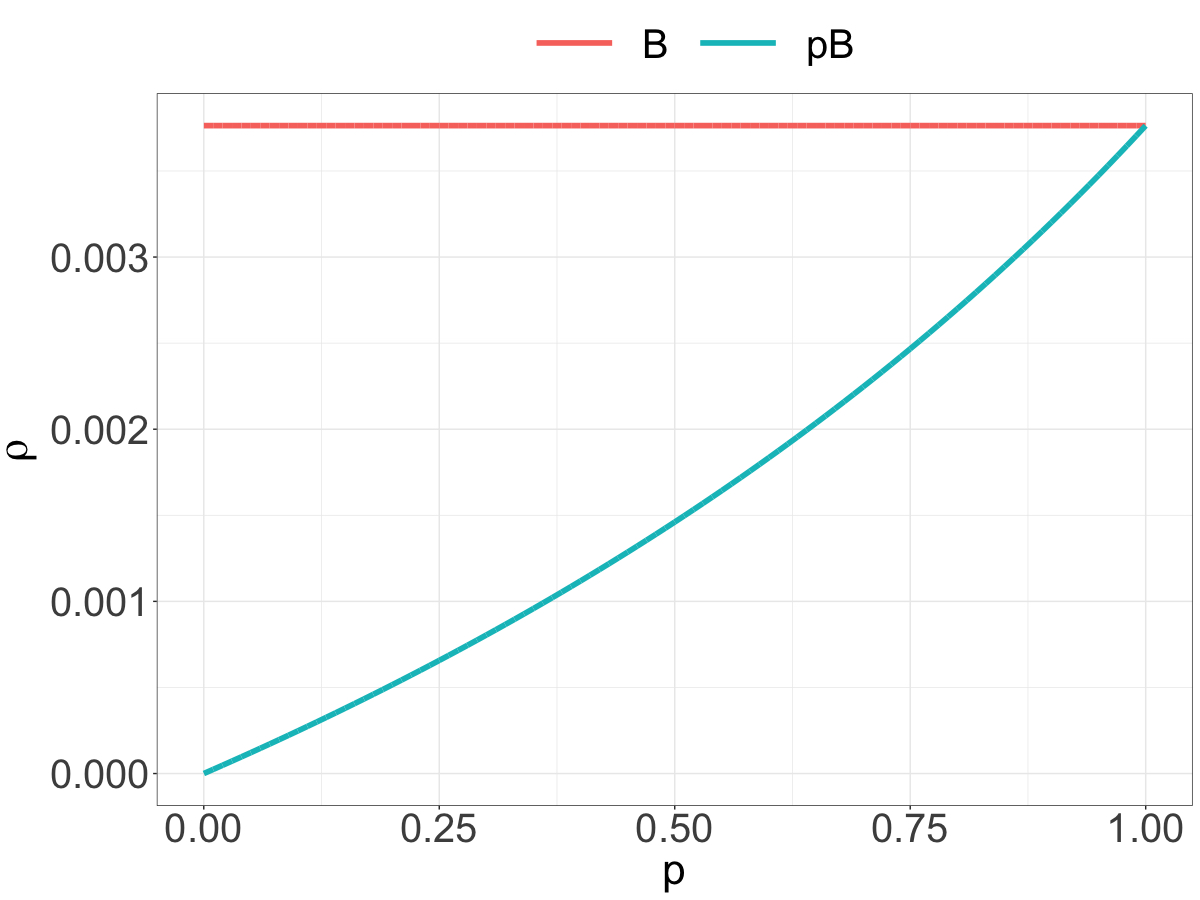}
	}
	\hfil
	\subfigure[unbalanced: $ \boldsymbol{\pi} = (\frac{1}{8}, \frac{1}{8}, \frac{3}{8}, \frac{3}{8}) $ \label{fig:rho0b}]{
		\includegraphics[width=0.45\textwidth]{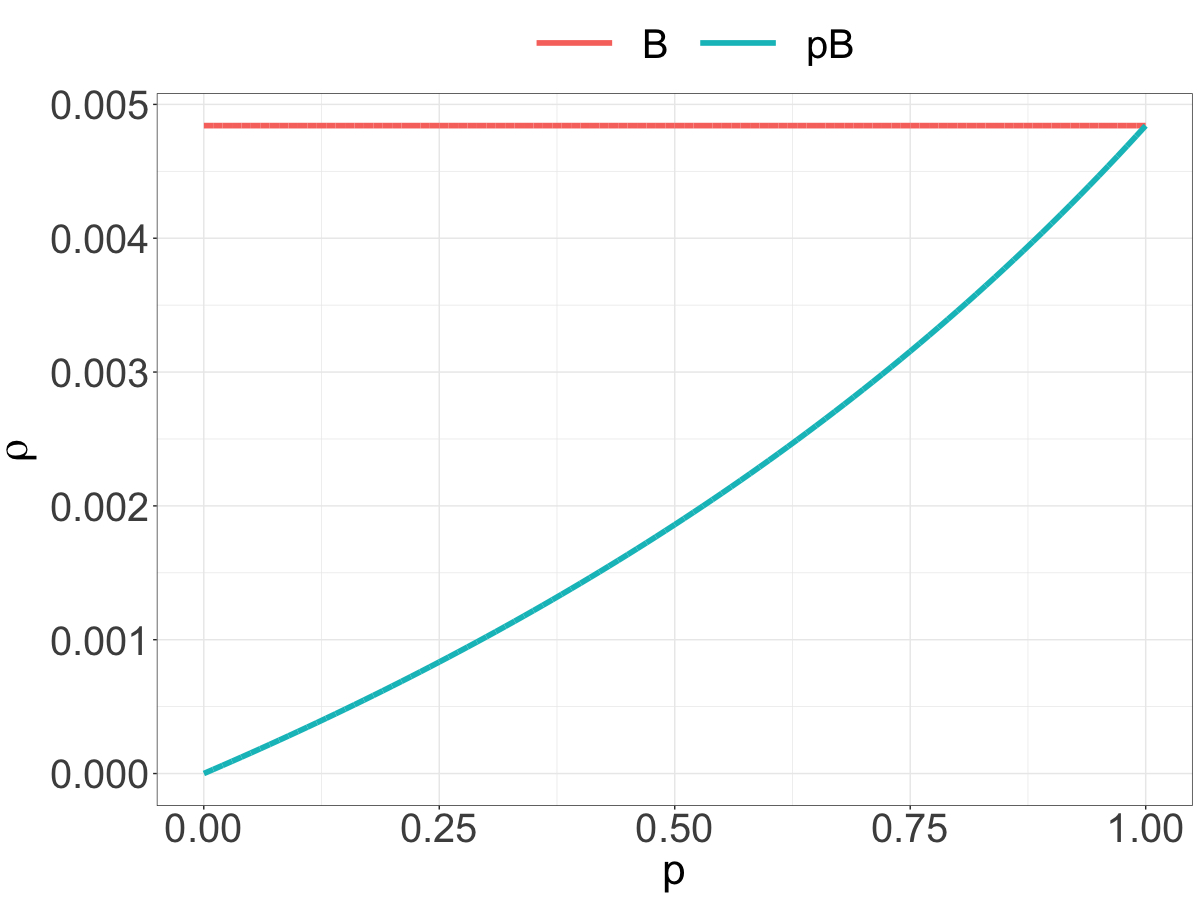}
	}
	\caption{Chernoff information $ \rho $ as in Eq.~\eqref{eq:rhoapprox} corresponding to $ \mathbf{B} $ as in Eq.~\eqref{eq:exampleB} and $ p \mathbf{B} $ for $ p \in  (0, 1) $.}
	\label{fig:rho0}
\end{figure}

Now given dynamic network sampling parameter $ p_1 \in (0, 1-p_0) $, the baseline sampling scheme can proceed uniformly at random again, i.e., 
\begin{equation}
\label{eq:B1}
\mathbf{B}_1 = \mathbf{B}_0 + p_1 \mathbf{B} = (p_0 + p_1) \mathbf{B}.
\end{equation}

This dynamic network sampling simulates the situation when one is given limited resources to sample some extra edges after observing the partial graph with only a small portion of the edges. Since we only have limited budget to sample another small portion of edges, one would benefit from identifying vertex pairs that have much influence on the community structure. In other words, the baseline sampling scheme just randomly choosing vertex pairs without using the information from the initial observed graphs and our goal is to design an alternative scheme to optimize this dynamic network sampling procedure so that one could have a better block recovery even with limited resources to only observe a partial graph with a small portion of the edges.

\begin{corollary}
	\label{cor:Chernoff-Superiority}
	For $K$-block SBMs, given block connectivity probability matrix $ \mathbf{B} \in (0, 1)^{K \times K} $ and a vector of block assignment probabilities $ \boldsymbol{\pi} \in (0, 1)^K $. We have $ \mathbf{B} \succ \mathbf{B}_1 \succ \mathbf{B}_0 $ where $ \mathbf{B}_0 $ is defined as in Eq.~\eqref{eq:B0} with $ p_0 \in (0, 1) $ and $ \mathbf{B}_1 $ is defined as in Eq.~\eqref{eq:B1} with $ p_1 \in (0, 1-p_0) $.
\end{corollary}

The proof of Corollary~\ref{cor:Chernoff-Superiority} can be found in Appendix. This corollay implies that we can have a better block recovery from $ \mathbf{B}_1 $ than $ \mathbf{B}_0 $.

\begin{assumption}
	\label{cond:1}
	The Chernoff-active blocks after initial sampling is unique, i.e., there exists an unique pair $ \left(k_0^*, \ell_0^* \right) \in \{(k, \ell) \; | \; 1 \leq k < \ell \leq K \} $ such that
	\begin{equation}
	\left(k_0^*, \ell_0^* \right) = \arg \min_{k \neq l} C_{k ,\ell}(\mathbf{B}_0, \boldsymbol{\pi}),
	\end{equation}
	where $ \mathbf{B}_0 $ is defined as in Eq.~\eqref{eq:B0} and $ \boldsymbol{\pi} $ is the vector of block assignment probabilities.
\end{assumption}

To improve this baseline sampling scheme, we concentrate on the Chernoff-active blocks $ \left(k_0^*, \ell_0^* \right) $ after initial sampling assuming Assumption~\ref{cond:1} holds. Instead of sampling from the entire block connectivity probability matrix $ \mathbf{B} $ like the baseline sampling scheme as in Eq.~\eqref{eq:B1}, we only sample the entries associated with the Chernoff-active blocks. As a competitor to $ \mathbf{B}_1 $, our Chernoff-optimal dynamic network sampling scheme is then given by
\begin{equation}
\label{eq:B1tilde}
\widetilde{\mathbf{B}}_1 = \mathbf{B}_0 + \frac{p_1}{\left(\pi_{k_0^*} + \pi_{\ell_0^*}\right)^2 } \mathbf{B} \circ \mathbf{1}_{k_0^*, \ell_0^*},
\end{equation}
where $ \circ $ denotes Hadamard product,  $ \pi_{k_0^*}  $ and $ \pi_{\ell_0^*} $ denote the block assignment probabilities for block $ k_0^* $ and $ \ell_0^* $ respectively, and $ \mathbf{1}_* $ is the $ K \times K $ binary matrix with 0's everywhere except for 1's associated with the Chernoff-active blocks $ \left(k_0^*, \ell_0^* \right) $, i.e., for any $ i, j \in \{1, \cdots, K \} $
\begin{equation}
\mathbf{1}_{k_0^*, \ell_0^*}[i, j]  = 
\begin{cases}
1 & \text{if} \;\; (i, j) \in \left\{ \left(k_0^*, k_0^* \right), \;  \left(k_0^*, \ell_0^* \right), \;  \left(\ell_0^*, k_0^* \right), \;  \left(\ell_0^*, \ell_0^* \right)  \right\} \\
0 & \text{otherwise}
\end{cases}
.
\end{equation}

Note that the multiplier $ \frac{1}{\left(\pi_{k_0^*} + \pi_{\ell_0^*}\right)^2} $ on $ p_1 \mathbf{B} \circ \mathbf{1}_* $ assures that we sample the same number of potential edges with $ \widetilde{\mathbf{B}}_1 $ as we do with $ \mathbf{B}_1 $ in the baseline sampling scheme. In addition, to avoid over-sampling with respect to $ \mathbf{B} $, i.e., to ensure $ \widetilde{\mathbf{B}}_1[i, j] \leq \mathbf{B}[i, j] $ for any $ i, j \in \{1, \cdots, K \} $, we require
\begin{equation}
\label{eq:p1max}
p_1 \leq p_1^{\text{max}} = \left( 1 - p_0 \right) \left(\pi_{k_0^*} + \pi_{\ell_0^*}\right)^2.
\end{equation} 

\begin{assumption}
	\label{cond:Chernoff-Superiority2}
	For $K$-block SBMs, given a block connectivity probability matrix $ \mathbf{B} \in (0, 1)^{K \times K} $ and a vector of block assignment probabilities $ \boldsymbol{\pi} \in (0, 1)^K $. Let $ p_1^* \in (0, p_1^{\text{max}}] $ be the smallest positive $ p_1 \leq p_1^{\text{max}} $ such that 
	\begin{equation}
	\arg \min_{k \neq l} C_{k ,\ell}(\widetilde{\mathbf{B}}_1, \boldsymbol{\pi})
	\end{equation}
	is not unique where $ p_1^{\text{max}} $ is defined as in Eq.~\eqref{eq:p1max} and $ \widetilde{\mathbf{B}}_1 $ is defined as in Eq.~\eqref{eq:B1tilde}. If the arg min is always unique, let $ p_1^* = p_1^{\text{max}} $.
\end{assumption}

For any $ p_1 \in (0, p_1^*)$, we can have a better block recovery from $ \widetilde{\mathbf{B}}_1 $ than $ \mathbf{B}_1 $, i.e., our Chernoff-optimal dynamic network sampling sheme is better than the baseline sampling scheme in terms of block recovery. 

As an illustaration, consider the 4-block SBM with initial sampling parameter $ p_0 = 0.01 $ and block connectivity probability matrix $ \mathbf{B} $ as in Eq.~\eqref{eq:exampleB}. Figure~\ref{fig:rho1} shows the Chernoff information $ \rho $ as in Eq.~\eqref{eq:rhoapprox} corresponding to $ \mathbf{B} $ as in Eq.~\eqref{eq:exampleB}, $ \mathbf{B}_0 $ as in Eq.~\eqref{eq:B0}, $ \mathbf{B}_1 $ as in Eq.~\eqref{eq:B1}, and $ \widetilde{\mathbf{B}}_1 $ as in Eq.~\eqref{eq:B1tilde} with dynamic network sampling parameter $ p_1 \in (0, p_1^*) $ where $ p_1^* $ is defined as in Assumption~\ref{cond:Chernoff-Superiority2}. In addition, Figure~\ref{fig:rho1a} assumes $ \boldsymbol{\pi} = (\frac{1}{4}, \frac{1}{4}, \frac{1}{4}, \frac{1}{4}) $ and Figure~\ref{fig:rho1b} assumes $ \boldsymbol{\pi} = (\frac{1}{8}, \frac{1}{8}, \frac{3}{8}, \frac{3}{8}) $. Note that for any $ p_1 \in (0, p_1^*) $ we have $\rho_{B} > \rho_{\widetilde{B}_1} > \rho_{B_1} > \rho_{B_0} $ and thus $ \mathbf{B} \succ \widetilde{\mathbf{B}}_1 \succ \mathbf{B}_1 \succ \mathbf{B}_0 $. That is, in terms of Chernoff information, when given same amount of resources, the proposed Chernoff-optimal dynamic network sampling scheme can yield better block recovery results. In other words, to reach the same level of performance, in terms of Chernoff information, the proposed Chernoff-optimal dynamic network sampling scheme needs less resources.

\begin{figure}[h!]
	\subfigure[balanced: $ \boldsymbol{\pi} = (\frac{1}{4}, \frac{1}{4}, \frac{1}{4}, \frac{1}{4}) $ \label{fig:rho1a}]{
		\includegraphics[width=0.45\textwidth]{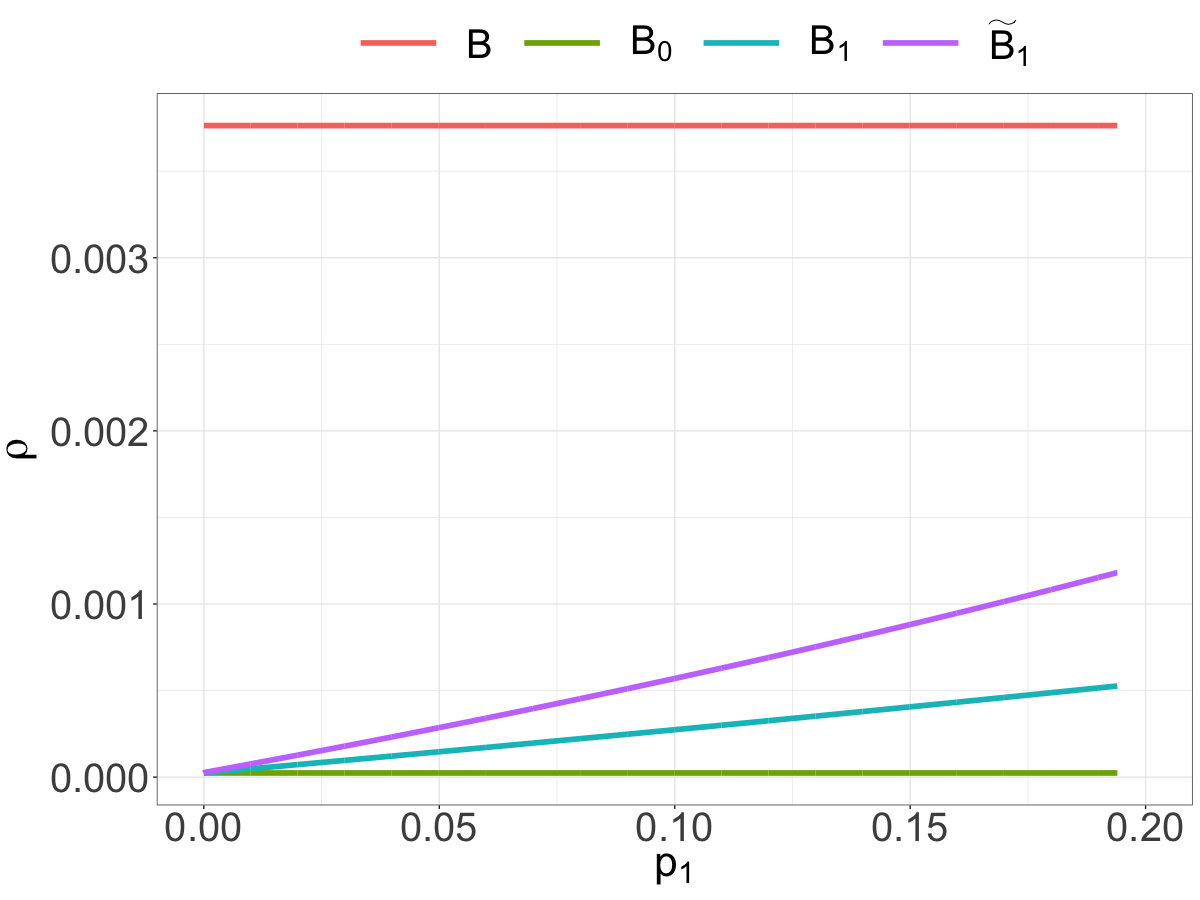}
	}
	\hfil
	\subfigure[unbalanced: $ \boldsymbol{\pi} = (\frac{1}{8}, \frac{1}{8}, \frac{3}{8}, \frac{3}{8}) $ \label{fig:rho1b}]{
		\includegraphics[width=0.45\textwidth]{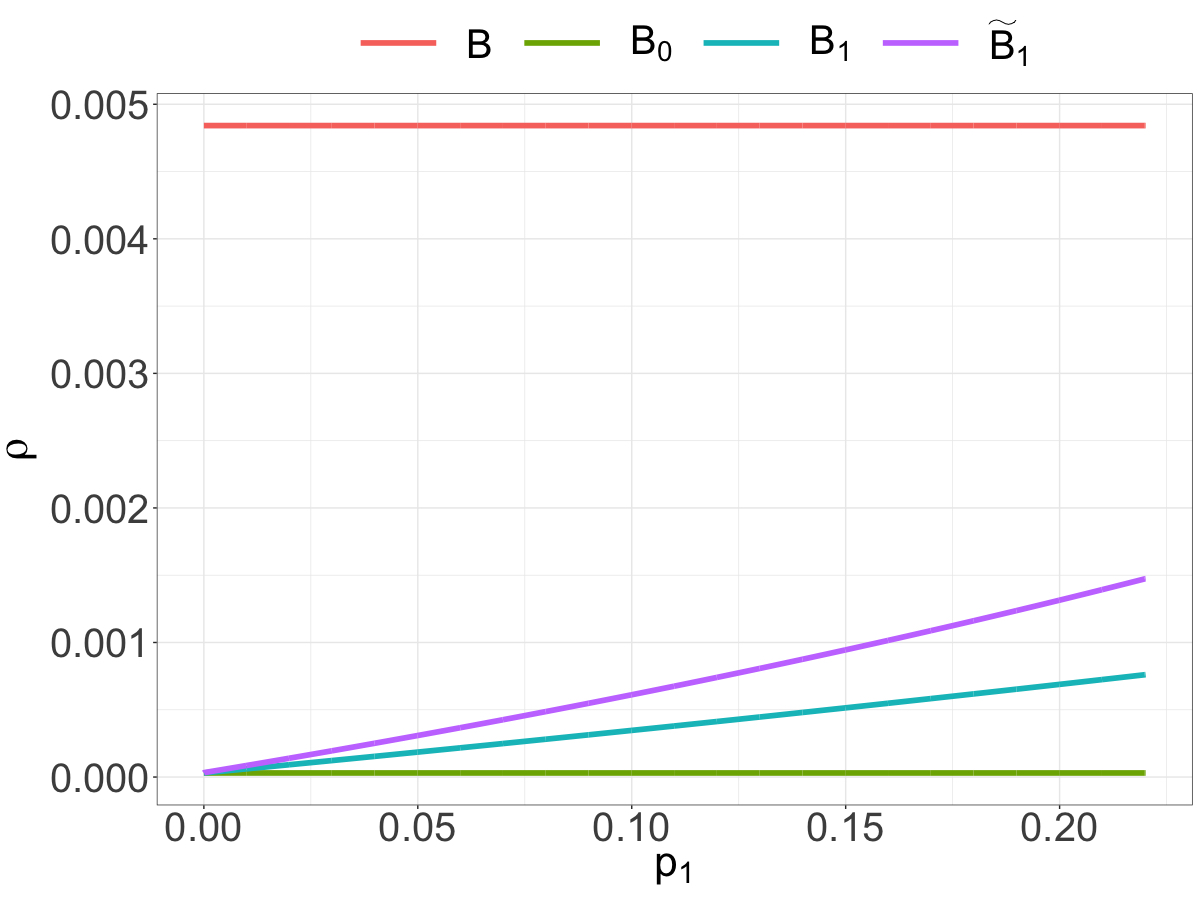}
	}
	\caption{Chernoff information $ \rho $ as in Eq.~\eqref{eq:rhoapprox} corresponding to $ \mathbf{B} $ as in Eq.~\eqref{eq:exampleB}, $ \mathbf{B}_0 $ as in Eq.~\eqref{eq:B0}, $ \mathbf{B}_1 $ as in Eq.~\eqref{eq:B1}, and $ \widetilde{\mathbf{B}}_1 $ as in Eq.~\eqref{eq:B1tilde} with initial sampling parameter $ p_0 = 0.01 $ and dynamic network sampling parameter $ p_1 \in (0, p_1^*) $ where $ p_1^* $ is defined as in Assumption~\ref{cond:Chernoff-Superiority2}.}
	\label{fig:rho1}
\end{figure}

As described earlier, it may be the case that $ p_1^* < p_1^{\text{max}} $ at which point Chernoff-active blocks change to $ (k_1^*, \ell_1^*) $. This potential non-uniquess of the Chernoff argmin is a consequence of our dynamic network sampling scheme. In the case of $ p_1 > p_1^* $, our Chernoff-optimal dynamic network sampling scheme is adopted as
\begin{equation}
\label{eq:B1tildestar}
\widetilde{\mathbf{B}}_1^* = \mathbf{B}_0 + \left(p_1 - p_1^* \right) \mathbf{B} + \frac{p_1^*}{\left(\pi_{k_0^*} + \pi_{\ell_0^*}\right)^2 } \mathbf{B} \circ \mathbf{1}_{k_0^*, \ell_0^*},
\end{equation}

Similarly, the multiplier $ \frac{1}{\left(\pi_{k_0^*} + \pi_{\ell_0^*}\right)^2} $ on $ p_1^* \mathbf{B} \circ \mathbf{1}_{k_0^*, \ell_0^*} $ assures that we sample the same number of potential edges with $ \widetilde{\mathbf{B}}_1^* $ as we do with $ \mathbf{B}_1 $ in the baseline sampling scheme. In addition, to avoid over-sampling with respect to $ \mathbf{B} $, i.e., $ \widetilde{\mathbf{B}}_1^*[i, j] \leq \mathbf{B}[i, j] $ for any $ i, j \in \{1, \cdots, K \} $, we require
\begin{equation}
\label{eq:p11max}
p_1 \leq p_{11}^{\text{max}} = 1 - p_0 - \frac{p_1^*}{\left(\pi_{k_0^*} + \pi_{\ell_0^*}\right)^2 } + p_1^*.
\end{equation} 	

For any $ p_1 \in [p_1^*, p_{11}^{\text{max}}] $, we can have a better block recovery from $ \widetilde{\mathbf{B}}_1^* $ than $ \mathbf{B}_1 $, i.e., our Chernoff-optimal dynamic network sampling sheme is again better than the baseline sampling scheme in terms of block recovery. 

As an illustration, consider a 4-block SBM with initial sampling parameter $ p_0 = 0.01 $ and block connectivity probability matrix $ \mathbf{B} $ as in Eq.~\eqref{eq:exampleB}. Figure~\ref{fig:rho2} shows the Chernoff information $ \rho $ as in Eq.~\eqref{eq:rhoapprox} corresponding to $ \mathbf{B} $ as in Eq.~\eqref{eq:exampleB}, $ \mathbf{B}_0 $ as in Eq.~\eqref{eq:B0}, $ \mathbf{B}_1 $ as in Eq.~\eqref{eq:B1}, and $ \widetilde{\mathbf{B}}_1^* $ as in Eq.~\eqref{eq:B1tildestar} with dynamic network sampling parameter $ p_1 \in [p_1^*, p_{11}^{\text{max}}] $ where $ p_1^* $ is defined as in Assumption~\ref{cond:Chernoff-Superiority2} and $ p_{11}^{\text{max}} $ is defined as in Eq.~\eqref{eq:p11max}. In addition, Figure~\ref{fig:rho2a} assumes $ \boldsymbol{\pi} = (\frac{1}{4}, \frac{1}{4}, \frac{1}{4}, \frac{1}{4}) $ and Figure~\ref{fig:rho2b} assumes $ \boldsymbol{\pi} = (\frac{1}{8}, \frac{1}{8}, \frac{3}{8}, \frac{3}{8}) $. Note that for any $ p_1 \in [p_1^*, p_{11}^{\text{max}}] $ we have $\rho_{B} > \rho_{\widetilde{B}_1^*} > \rho_{B_1} > \rho_{B_0} $ and thus $ \mathbf{B} \succ \widetilde{\mathbf{B}}_1^* \succ \mathbf{B}_1 \succ \mathbf{B}_0 $. That is, the adopted Chernoff-optimal dynamic network sampling scheme can still yield better block recovery results, in terms of Chernoff information, given the same amout of resources.

\begin{figure}[h!]
	\subfigure[balanced: $ \boldsymbol{\pi} = (\frac{1}{4}, \frac{1}{4}, \frac{1}{4}, \frac{1}{4}) $ \label{fig:rho2a}]{
		\includegraphics[width=0.45\textwidth]{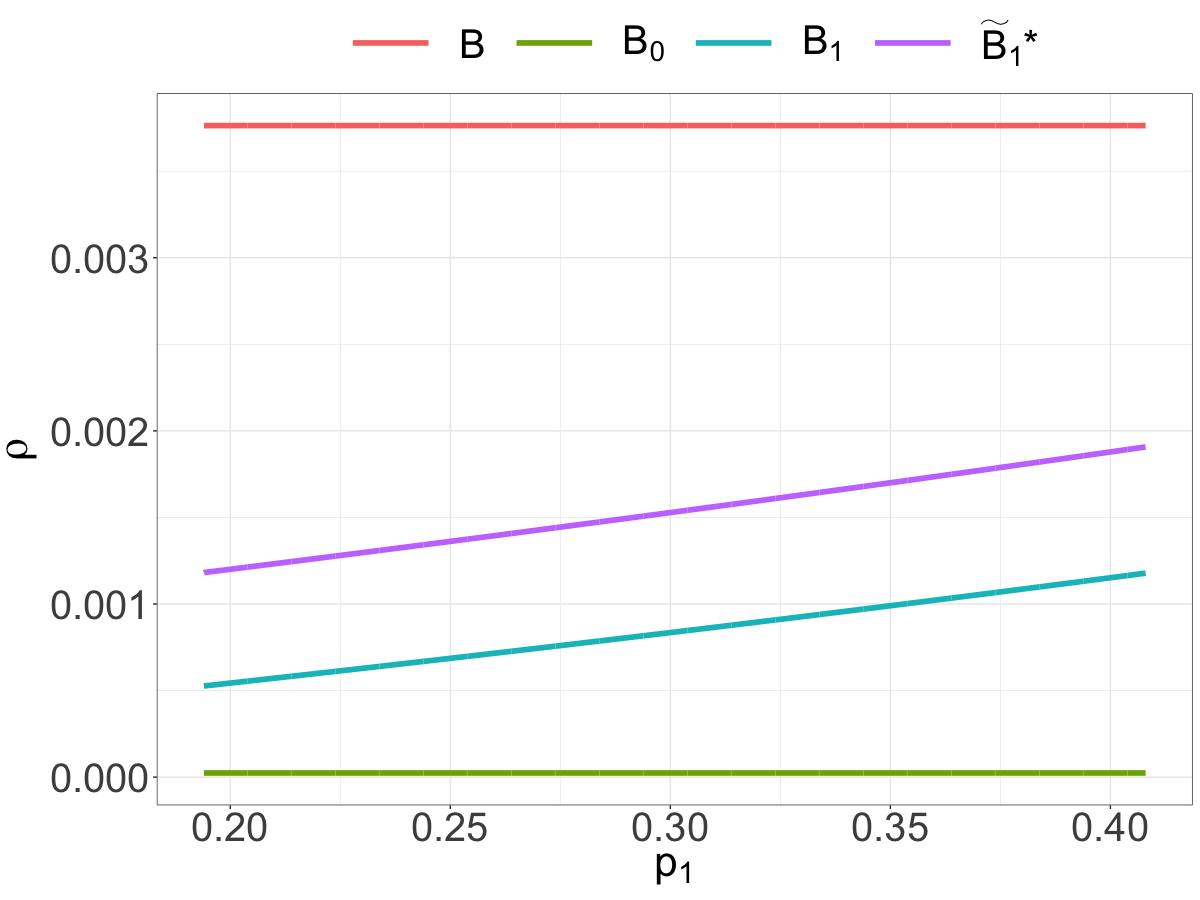}
	}
	\hfil
	\subfigure[unbalanced: $ \boldsymbol{\pi} = (\frac{1}{8}, \frac{1}{8}, \frac{3}{8}, \frac{3}{8}) $ \label{fig:rho2b}]{
		\includegraphics[width=0.45\textwidth]{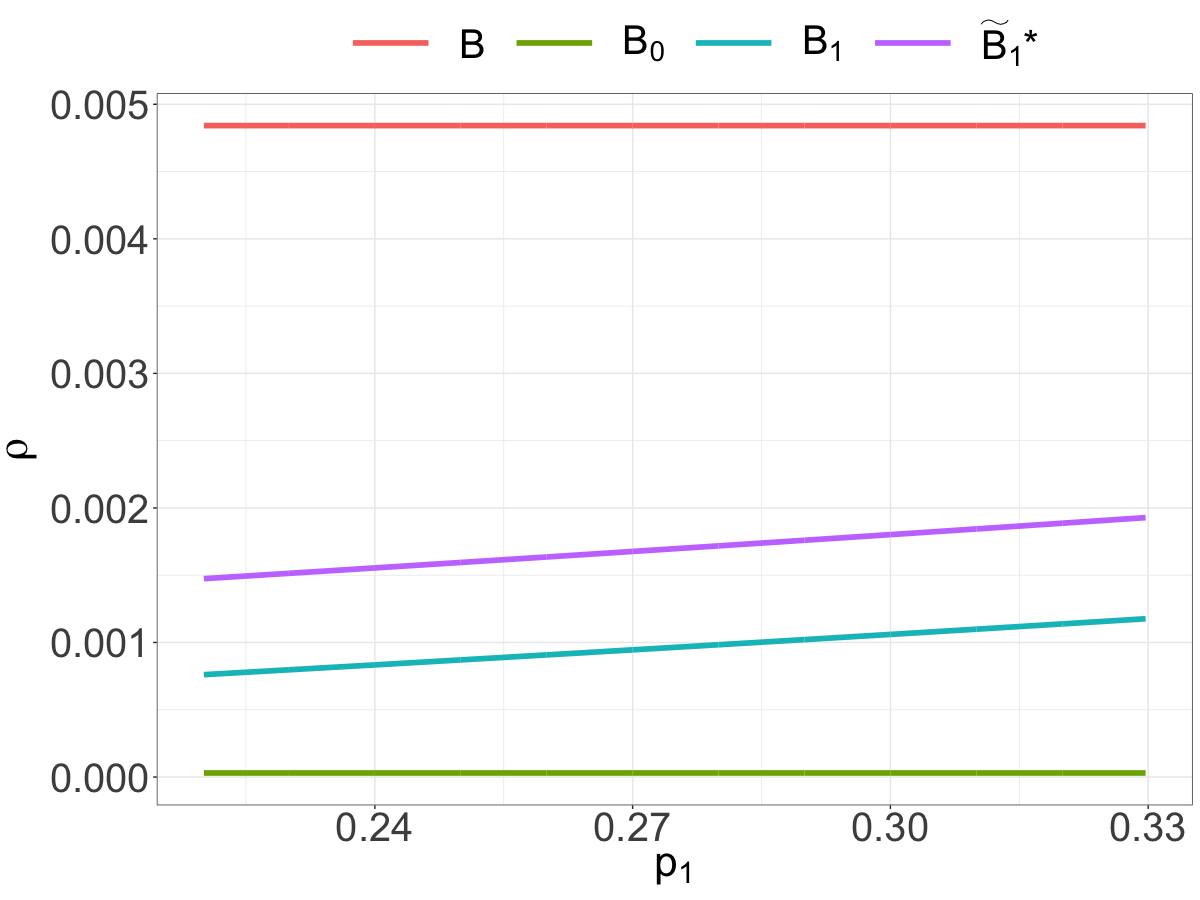}
	}
	\caption{Chernoff information $ \rho $ as in Eq.~\eqref{eq:rhoapprox} corresponding to $ \mathbf{B} $ as in Eq.~\eqref{eq:exampleB}, $ \mathbf{B}_0 $ as in Eq.~\eqref{eq:B0}, $ \mathbf{B}_1 $ as in Eq.~\eqref{eq:B1}, and $ \widetilde{\mathbf{B}}_1^* $ as in Eq.~\eqref{eq:B1tildestar} with initial sampling parameter $ p_0 = 0.01 $ and dynamic network sampling parameter $ p_1 \in [p_1^*, p_{11}^{\text{max}}] $ where $ p_1^* $ is defined as in Assumption~\ref{cond:Chernoff-Superiority2} and $ p_{11}^{\text{max}} $ is defined as in Eq.~\eqref{eq:p11max}.}
	\label{fig:rho2}
\end{figure}

Now we illustrate how the proposed Chernoff-optimal dynamic network sampling sheme can be migrated for real applications. We summarize the uniform dynamic sampling scheme (baseline) as Algorithm~\ref{algo:1} and our Chernoff-optimal dynamic network sampling scheme as Algorithm~\ref{algo:2}. Recall given potential edge set $ E $ and initial sampling parameter $ p_0 \in (0, 1) $, we have the initial edge set $ E_0 \subset E $ with $ \lvert E_0 \rvert = p_0 \lvert E \rvert $. The goal is to dynamically sample new edges from the potential edge set so that we can have a better block recovery given limited resources.

\begin{algorithm}
	\label{algo:1}
	\SetAlgoNoLine
	\KwIn{Number of vertices $ n $; potential edge set $ E = \{(i, j) \; | \; i, j \in \{1, \cdots, n \} \} $; initial edge set $ E_0 \subset E $; dynamic network sampling parameter $ p_1 \in \left(0, 1- \frac{\lvert E_0 \rvert}{\lvert E \rvert} \right) $}
	
	Construct dynamic edge set as
	\begin{equation*}
	E_1 = \left\{(i ,j) \; | \; (i ,j) \in E \setminus E_0 \right\} \qquad \text{with} \qquad \lvert E_1 \rvert = p_1 \lvert E \rvert.
	\end{equation*} \\
	
	Construct dynamic adjacency matrix as $ \mathbf{A} \in \{0, 1\}^{n \times n} $ where for any $ i, j \in \{1, \cdots, n \} $
	\begin{equation*}
	\mathbf{A}[i, j]  = 
	\begin{cases}
	1 & \text{if} \;\; (i, j) \in E_0 \bigcup E_1 \;\; \text{or} \;\; (j, i) \in E_0 \bigcup E_1 \\
	0 & \text{otherwise}
	\end{cases}
	.
	\end{equation*} \\
	
	Estimate dynamic latent positions as $ \mathbf{\widehat{X}} \in \mathbb{R}^{n \times \widehat{d}} $ using ASE of $ \mathbf{A} $ where $ \widehat{d} $ is chosen as in Remark~\ref{remark:dhat}. \\
	
	Cluster $ \mathbf{\widehat{X}} $ using Gaussian mixture modeling (GMM) to estimate the block assignments as $ \boldsymbol{\widehat{\tau}} \in \{1, \cdots, \widehat{K} \}^{n} $ where $ \widehat{K} $ is chosen via Bayesian Information Criterion (BIC).
	
	\KwOut{Block assignments $ \boldsymbol{\widehat{\tau}} $.}
	
	\caption{Uniform dynamic network sampling scheme (baseline)}
\end{algorithm}

\begin{algorithm}
	\label{algo:2}
	\SetAlgoNoLine
	\KwIn{Number of vertices $ n $; potential edge set $ E = \{(i, j) \; | \; i, j \in \{1, \cdots, n \} \} $; initial edge set $ E_0 \subset E $; dynamic network sampling parameter $ p_1 \in \left(0, 1- \frac{\lvert E_0 \rvert}{\lvert E \rvert} \right) $}
	
	Construct dynamic adjacency matrix as $ \mathbf{A} \in \{0, 1\}^{n \times n} $ where for any $ i, j \in \{1, \cdots, n \} $
	\begin{equation*}
	\mathbf{A}[i, j]  = 
	\begin{cases}
	1 & \text{if} \;\; (i, j) \in E_0 \;\; \text{or} \;\; (j, i) \in E_0 \\
	0 & \text{otherwise}
	\end{cases}
	.
	\end{equation*} \\
	
	Estimate dynamic latent positions as $ \mathbf{\widehat{X}} \in \mathbb{R}^{n \times \widehat{d}} $ using ASE of $ \mathbf{A} $ where $ \widehat{d} $ is chosen as in Remark~\ref{remark:dhat}. \\
	
	Cluster $ \mathbf{\widehat{X}} $ using GMM to estimate the initial block assignments as $ \boldsymbol{\widehat{\xi}} \in \{1, \cdots, \widehat{K} \}^{n} $ where $ \widehat{K} $ is chosen via BIC. \\
	
	Estimate the dynamic block assignment probability vector as $ \boldsymbol{\widehat{\pi}} \in (0, 1)^K $ where for $ k \in \{1, \cdots, K \} $
	\begin{equation*}
	\widehat{\pi}_k = \frac{1}{n} \sum_{i=1}^{n} \mathbf{1} \{\boldsymbol{\widehat{\xi}}_i = k \}.
	\end{equation*} \\
	
	Estimate the dynamic block connectivity probability matrix as
	\begin{equation*}
	\mathbf{\widehat{B}} = \boldsymbol{\widehat{\mu}} \mathbf{I}_{d_+ d_-} \boldsymbol{\widehat{\mu}}^\top \in [0,1]^{\widehat{K} \times \widehat{K}},
	\end{equation*}
	where $ \boldsymbol{\widehat{\mu}} \in \mathbb{R}^{\widehat{K} \times \widehat{d}} $ is the estimated means of all clusters. \\
	
	Find the Chernoff-active blocks as 
	\begin{equation*}
	\left(k^*, \ell^* \right) = \arg \min_{k \neq l} C_{k ,\ell} \left(\mathbf{\widehat{B}}, \boldsymbol{\widehat{\pi}} \right).
	\end{equation*} \\
	
	Construct dynamic edge set as
	\begin{equation*}
	\begin{split}
	E_1 \subseteq E_* \qquad & \text{with} \qquad \lvert E_1 \rvert = \min \left\{p_1 \lvert E \rvert \left(\widehat{\pi}_{k^*} + \widehat{\pi}_{\ell^*}\right)^2, \lvert E_* \rvert \right\}, \\
	E_{11} \subset E \setminus \left(E_0 \bigcup E_1 \right) \qquad & \text{with} \qquad \lvert E_{11} \rvert = p_1 \lvert E \rvert - \lvert E_1 \rvert,
	\end{split}
	\end{equation*}
	where 
	\begin{equation*}
	E_* = \left\{(i ,j) \; | \; (i ,j) \in E \setminus E_0 \; \text{and} \; \widehat{\xi}_i, \widehat{\xi}_j \in \{k^*, \ell^* \} \right\}.
	\end{equation*} \\
	
	Update dynamic adjacency matrix as $ \mathbf{A} \in \{0, 1\}^{n \times n} $ where for any $ i, j \in \{1, \cdots, n \} $
	\begin{equation*}
	\mathbf{A}[i, j]  = 
	\begin{cases}
	1 & \text{if} \;\; (i, j) \in E_0 \bigcup E_1 \bigcup E_{11} \;\; \text{or} \;\; (j, i) \in E_0 \bigcup E_1 \bigcup E_{11} \\
	0 & \text{otherwise}
	\end{cases}
	.
	\end{equation*} \\
	
	Update dynamic latent positions as $ \mathbf{\widehat{X}} \in \mathbb{R}^{n \times \widehat{d}} $ using ASE of updated $ \mathbf{A} $ where $ \widehat{d} $ is chosen as in Remark~\ref{remark:dhat}. \\
	
	Cluster $ \mathbf{\widehat{X}} $ using GMM to estimate the block assignments as $ \boldsymbol{\widehat{\tau}} \in \{1, \cdots, \widehat{K} \}^{n} $ where $ \widehat{K} $ is chosen via BIC.
	
	\KwOut{Block assignments $ \boldsymbol{\widehat{\tau}} $.}
	
	\caption{Chernoff-optimal dynamic network sampling scheme}
\end{algorithm}

\section{Experiments}

\label{sec:5}

\subsection{Simulations}

In addition to Chernoff analysis, we also evalute our Chernoff-optimal dynamic network sampling sheme via simulations. In particular, consider the 4-block SBM parameterized by block connectivity probability matrix $ \mathbf{B} $ as in Eq.~\eqref{eq:exampleB} and dynamic network sampling parameter $ p_1 \in (0, p_{11}^{\text{max}}] $ where $ p_{11}^{\text{max}} $ is defined as in Eq.~\eqref{eq:p11max}. We fix initial sampling parameter $ p_0 = 0.01 $. For each $ p_1 \in (0, p_1^*) $ where $ p_1^* $ is defined as in Assumption~\ref{cond:Chernoff-Superiority2}, we simulate 50 adjacency matrices with $ n = 12000 $ vertices from $ \mathbf{B}_1 $ as in Eq.~\eqref{eq:B1} and $ \widetilde{\mathbf{B}}_1 $ as in Eq.~\eqref{eq:B1tilde} respectively. For each $ p_1 \in [p_1^*, p_{11}^{\text{max}}] $, we simulate 50 adjacency matrices with $ n = 12000 $ vertices from $ \mathbf{B}_1 $ as in Eq.~\eqref{eq:B1} and $ \widetilde{\mathbf{B}}_1^* $ as in Eq.~\eqref{eq:B1tildestar} respectively. In addition, Figure~\ref{fig:sim0a} assumes $ \boldsymbol{\pi} = (\frac{1}{4}, \frac{1}{4}, \frac{1}{4}, \frac{1}{4}) $, i.e., 3000 vertices in each block, and Figure~\ref{fig:sim0b} assumes $ \boldsymbol{\pi} = (\frac{1}{8}, \frac{1}{8}, \frac{3}{8}, \frac{3}{8}) $, i.e., 1500 vertices in two of the blocks and 4500 vertices in the other two blocks. We then apply ASE $ \circ $ GMM (Step 3 and 4 in Algorithm~\ref{algo:1}) to recover block assignments and adopt adjusted Rand index (ARI) to measure the performance. Figure~\ref{fig:sim0} shows ARI (\texttt{mean$ \pm $stderr}) associated with $ \mathbf{B}_1 $ for $ p_1 \in (0, p_{11}^{\text{max}}] $, $ \widetilde{\mathbf{B}}_1 $ for $ p_1 \in (0, p_1^*) $, and $ \widetilde{\mathbf{B}}_1^* $ for $ p_1 \in [p_1^*, p_{11}^{\text{max}}] $ where the dashed lines denote $ p_1^* $. Note that we can have a better block recovery from $ \widetilde{\mathbf{B}}_1 $ and $ \widetilde{\mathbf{B}}_1^* $ than $ \mathbf{B}_1 $, which argee with our results from Chernoff analysis.

\begin{figure}[h!]
	\subfigure[balanced: $ \boldsymbol{\pi} = (\frac{1}{4}, \frac{1}{4}, \frac{1}{4}, \frac{1}{4}) $ \label{fig:sim0a}]{
		\includegraphics[width=0.45\textwidth]{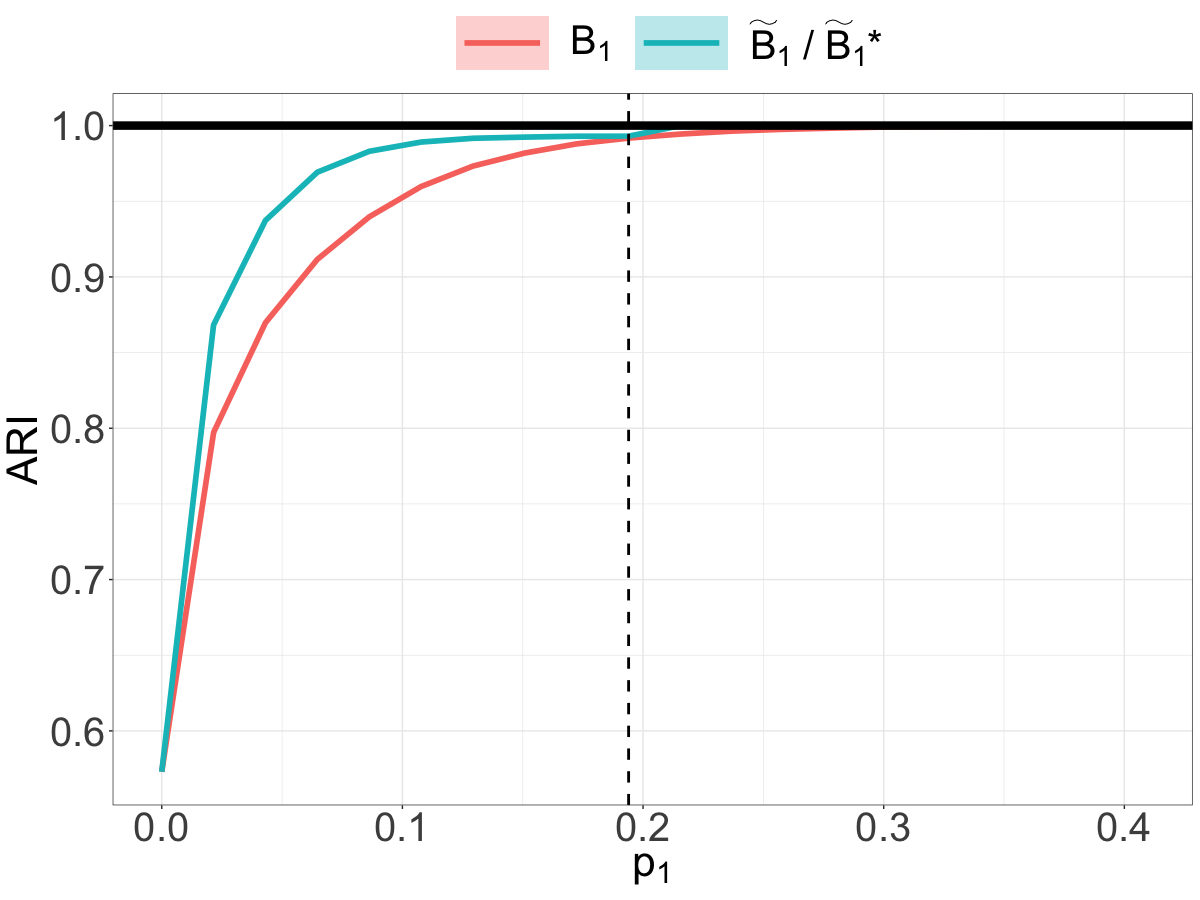}
	}
	\hfil
	\subfigure[unbalanced: $ \boldsymbol{\pi} = (\frac{1}{8}, \frac{1}{8}, \frac{3}{8}, \frac{3}{8}) $ \label{fig:sim0b}]{
		\includegraphics[width=0.45\textwidth]{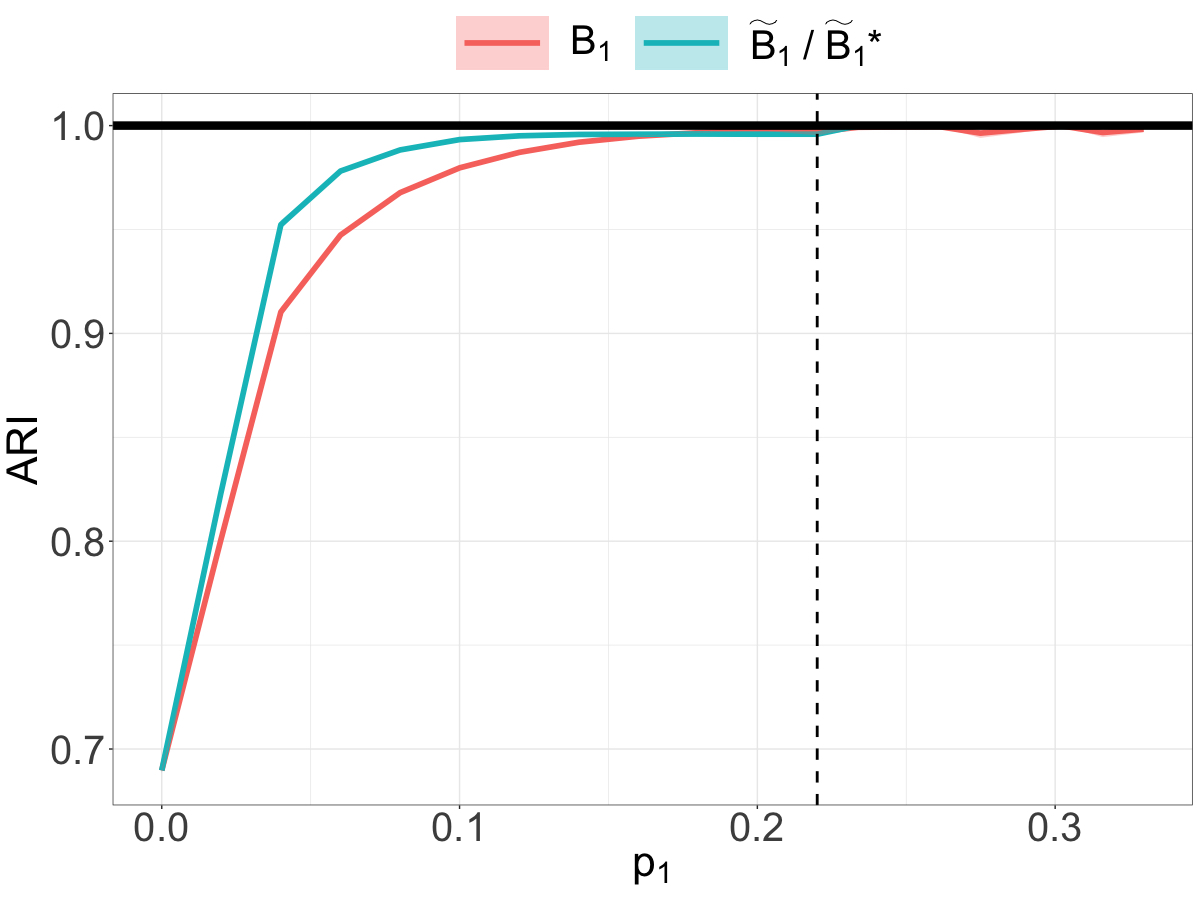}
	}
	\caption{Simulations for 4-block SBM parameterized by block connectivity probability matrix $ \mathbf{B} $ as in Eq.~\eqref{eq:exampleB} with initial sampling parameter $ p_0 = 0.01 $ and dynamic network sampling parameter $ p_1 \in (0, p_{11}^{\text{max}}] $ where $ p_{11}^{\text{max}} $ is defined as in Eq.~\eqref{eq:p11max}. The dashed lines denote $ p_1^* $ which is defined as in Assumption~\ref{cond:Chernoff-Superiority2}.}
	\label{fig:sim0}
\end{figure}

Now we compare the performance of Algorithms~\ref{algo:1} and~\ref{algo:2} by actual block recovery results. In particular, we start with the 4-block SBM parameterized by block connectivity probability matrix $ \mathbf{B} $ as in Eq.~\eqref{eq:exampleB}. We consider dynamic network sampling parameter $ p_1 \in (0, 1-p_0) $ where $ p_0 $ is the initial sampling parameter. For each $ p_1 $, we simulate 50 adjacency matrices with $ n = 4000 $ vertices and retrieve associated potential edge sets. We fix initial sampling parameter $ p_0 = 0.15 $ and randomly sample initial edge sets. We then apply both algorithms to estimate the block assignments and adopt ARI to measure the performance. Figure~\ref{fig:sim1} shows ARI (\texttt{mean$ \pm $stderr}) of two algorithms for $ p_1 \in (0, 0.85) $ where Figure~\ref{fig:sim1a} assumes $ \boldsymbol{\pi} = (\frac{1}{4}, \frac{1}{4}, \frac{1}{4}, \frac{1}{4}) $, i.e., 1000 vertices in each block, and Figure~\ref{fig:sim1b} assumes $ \boldsymbol{\pi} = (\frac{1}{8}, \frac{1}{8}, \frac{3}{8}, \frac{3}{8}) $, i.e., 500 vertices in two of the blocks and 1500 vertices in the other two blocks. Note that both algorithms tend to have a better performance as $ p_1 $ increases, i.e., as we sample more edges, and Algorithm~\ref{algo:2} can always recover more accurate block structure than Algorithm~\ref{algo:1}. That is, given the same amout of resources, the proposed Chernoff-optimal dynamic network sampling scheme can yield better block recovery results. In other words, to reach the same level of performance, in terms of the empirical clustering results, the proposed Chernoff-optimal dynamic network sampling scheme needs less resources.

\begin{figure}[h!]
	\subfigure[balanced: $ \boldsymbol{\pi} = (\frac{1}{4}, \frac{1}{4}, \frac{1}{4}, \frac{1}{4}) $ \label{fig:sim1a}]{
		\includegraphics[width=0.45\textwidth]{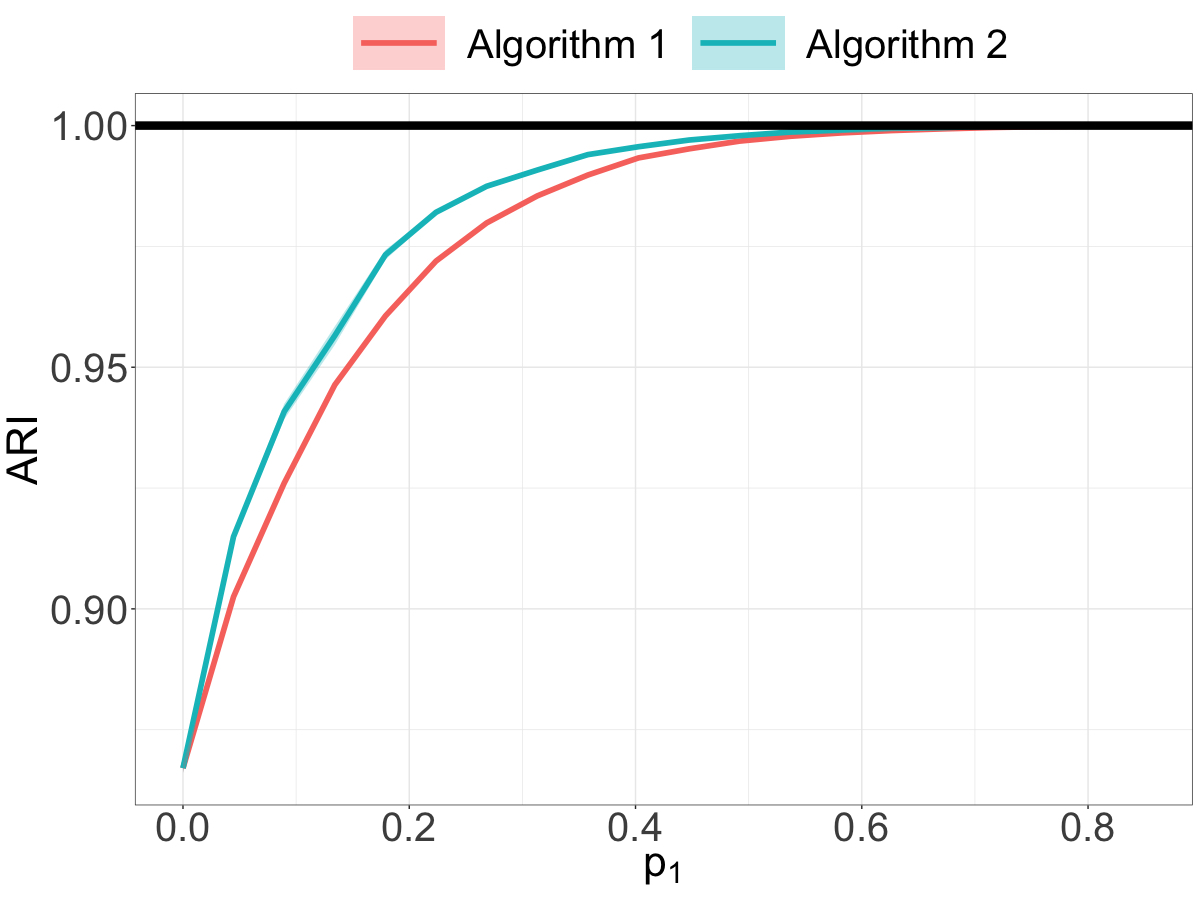}
	}
	\hfil
	\subfigure[unbalanced: $ \boldsymbol{\pi} = (\frac{1}{8}, \frac{1}{8}, \frac{3}{8}, \frac{3}{8}) $ \label{fig:sim1b}]{
		\includegraphics[width=0.45\textwidth]{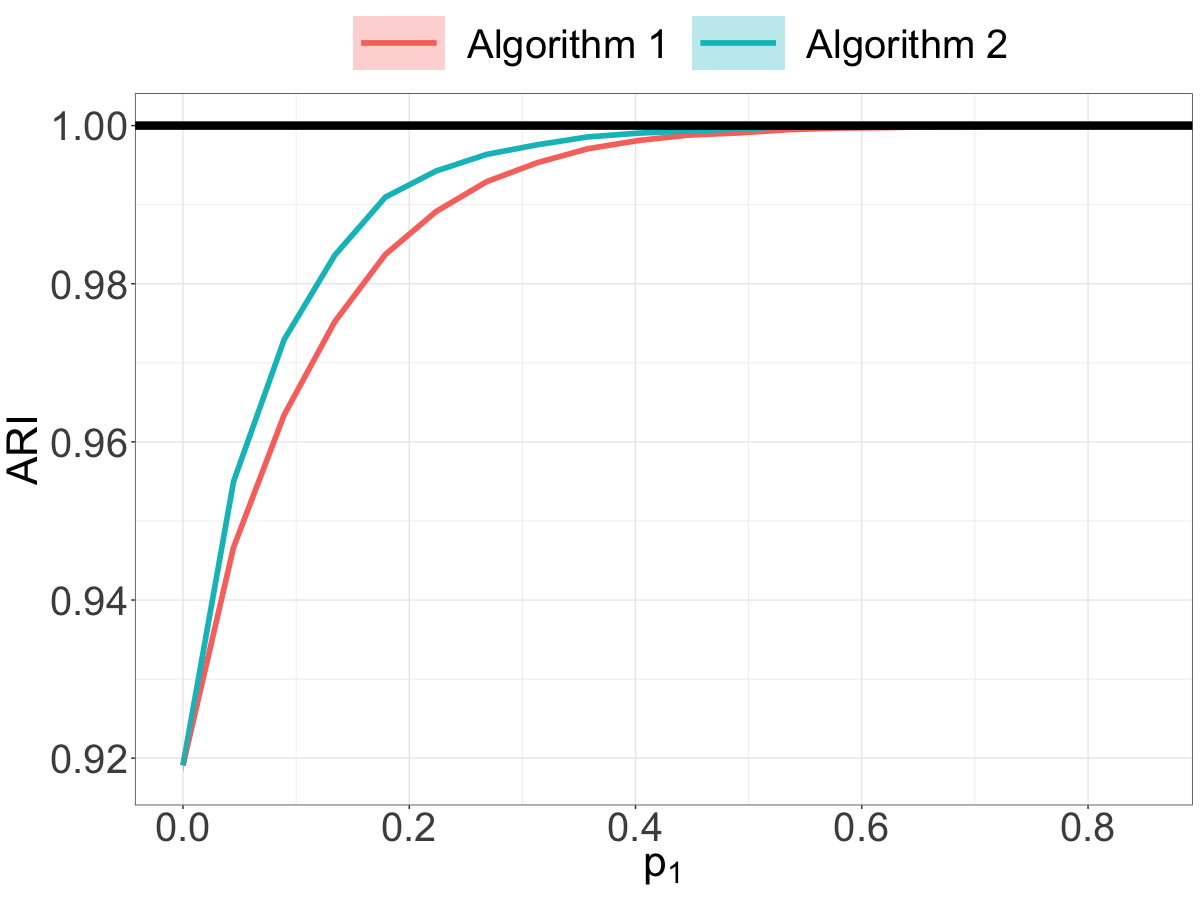}
	}
	\caption{Simulations for 4-block SBM parameterized by block connectivity probability matrix $ \mathbf{B} $ as in Eq.~\eqref{eq:exampleB} with initial sampling parameter $ p_0 = 0.15 $ and dynamic network sampling parameter $ p_1 \in (0, 0.85) $.}
	\label{fig:sim1}
\end{figure}	

\subsection{Real Data}

We also evaluate the performance of Algorithms~\ref{algo:1} and~\ref{algo:2} for real application. We conduct real data experiments on a diffusion MRI connectome dataset~\cite{Priebe2019}. There are 114 graphs (connectomes) estimated by the NDMG pipeline~\cite{Kiar2018} in this dataset. Each vertex in these graphs (the number of vertices $ n $ varies from 23728 to 42022) has a \{Left,~Right\} hemisphere label and a \{Gray,~White\} tissue label. We consider the potential 4 blocks as \{LG,~LW,~RG,~RW\} where L and R denote the Left and Right hemisphere label, G and W denote the Gray and White tissue label. Here we consider initial sampling parameter $ p_0 = 0.25 $ and dynamic network sampling parameter $ p_1 = 0.25 $. Let $ \Delta = \text{ARI(Algo2)} - \text{ARI(Algo1)} $ where ARI(Algo1) and ARI(Algo2) denotes the ARI when we apply Algorithms~\ref{algo:1} and~\ref{algo:2} respectively. The following hypothesis testing yields \texttt{p-value=0.0184}.
\begin{equation}
H_0: \; \text{median}(\Delta) \leq 0 \qquad \text{v.s.} \qquad H_A: \; \text{median}(\Delta) > 0.
\end{equation}

\begin{figure}[h!]
	\subfigure[Boxplot \label{fig:real1a}]{
		\includegraphics[width=0.45\textwidth]{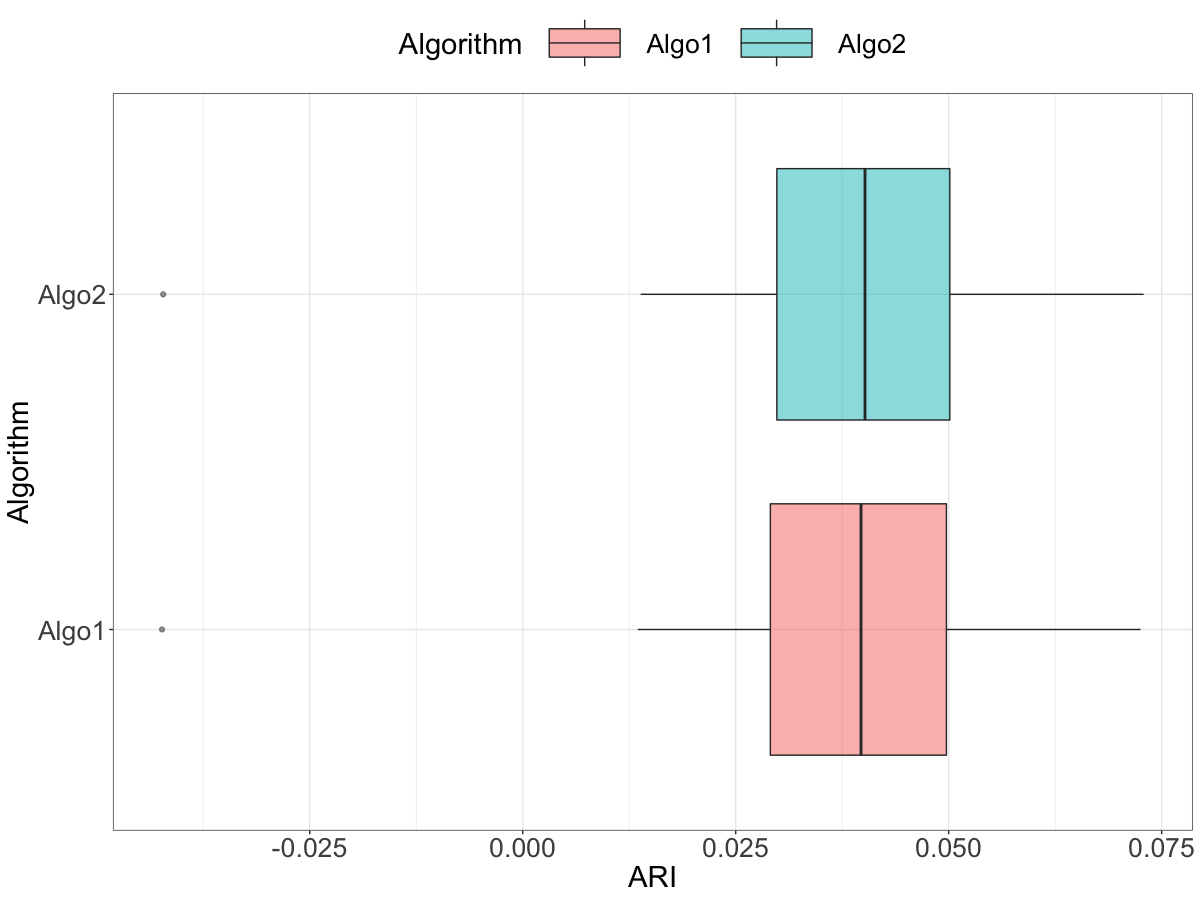}
	}
	\hfil
	\subfigure[Histogram \label{fig:real1b}]{
		\includegraphics[width=0.45\textwidth]{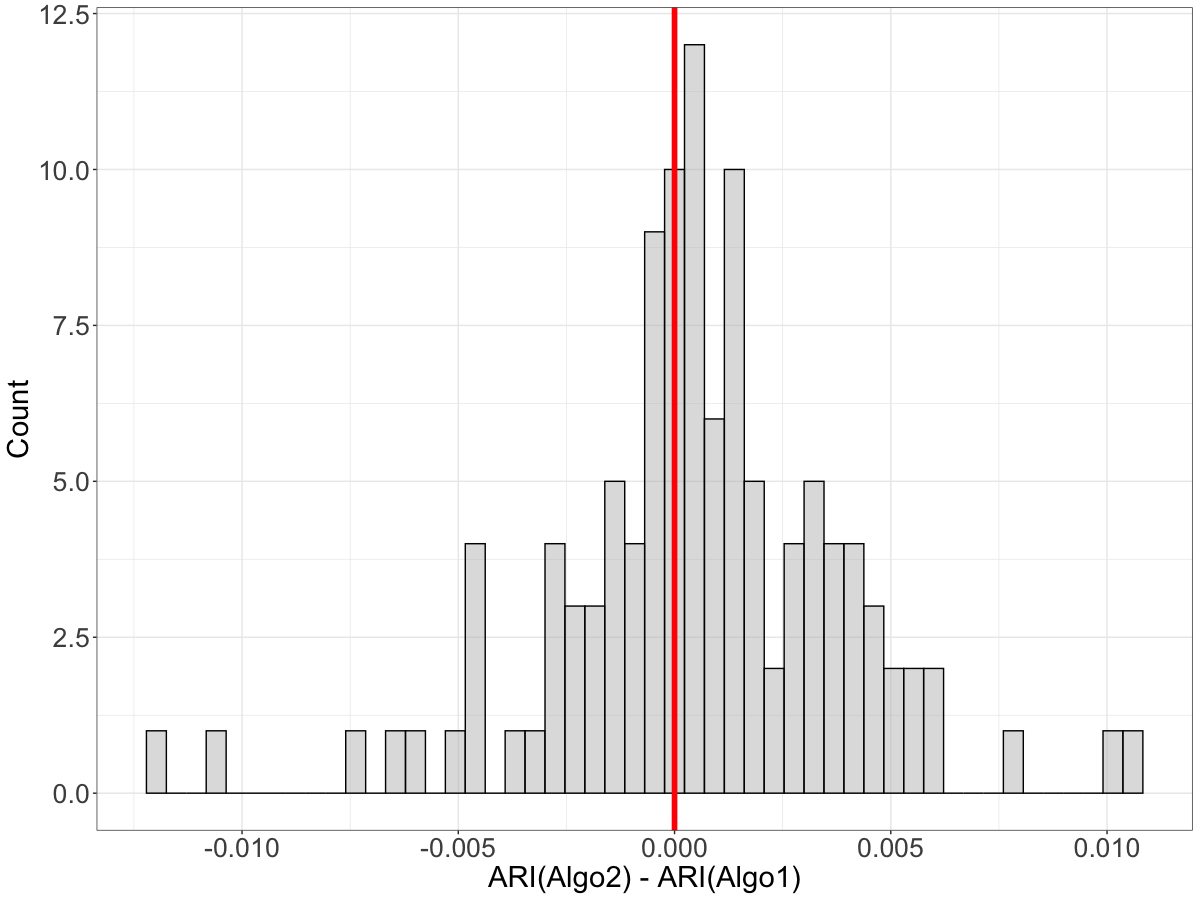}
	}
	\caption{Algorithms' comparative performance on diffusion MRI connectome data via ARI with initial sampling parameter $ p_0 = 0.25 $ and dynamic network sampling parameter $ p_1 = 0.25 $.}
	\label{fig:real1}
\end{figure}

Furthermore, we test our algorithms on a Microsoft bing entity dataset~\cite{Agterberg2020}. There are 2 graphs in this dataset where each has 13535 vertices. We treat block assignments estimated from the complete graph as ground truth. We consider initial sampling parameter $ p_0 \in \left\{0.2, \; 0.3 \right\} $ and dynamic network sampling parameter $ p_1 \in \left\{0, \; 0.05, \; 0.1, \; 0.15, \; 0.2 \right\} $. For each $ p_1 $, we sample 100 times and compare the overall performance of Algorithm~\ref{algo:1} and~\ref{algo:2}. Figure~\ref{fig:real2} shows the results where ARI is reported as \texttt{mean($\pm$stderr)}.

\begin{figure}[h!]
	\subfigure[$ p_0 = 0.2, \; p_1 \in \left\{0, \; 0.05, \; 0.1, \; 0.15, \; 0.2 \right\}  $ \label{fig:real2a}]{
		\includegraphics[width=0.45\textwidth]{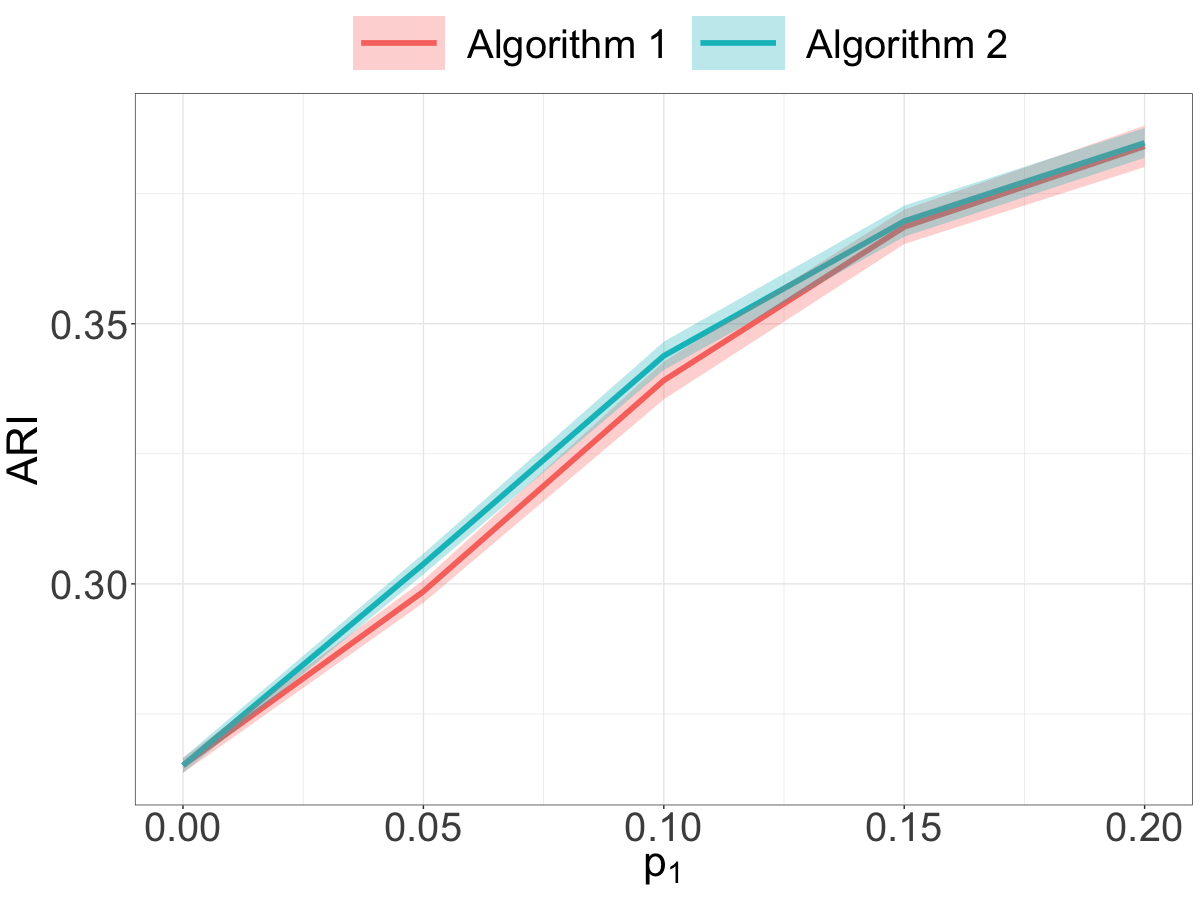}
	}
	\hfil
	\subfigure[$ p_0 = 0.3, \; p_1 \in \left\{0, \; 0.05, \; 0.1, \; 0.15, \; 0.2 \right\}  $ \label{fig:real2b}]{
		\includegraphics[width=0.45\textwidth]{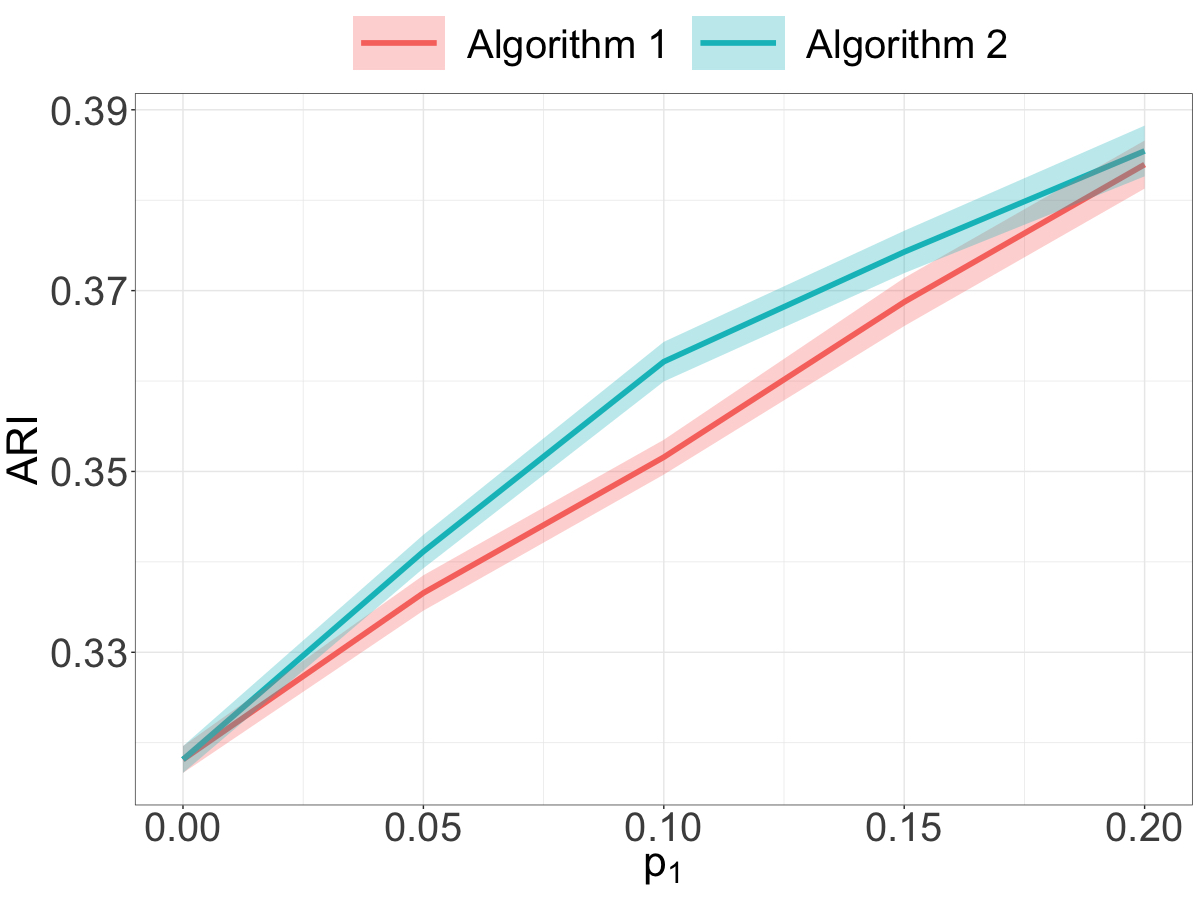}
	}
	\caption{Algorithms' comparative performance on Microsoft bing entity data via ARI with different initial sampling parameter $ p_0 $ and dynamic network sampling parameter $ p_1 $.}
	\label{fig:real2}
\end{figure}

We also conduct real data experiments with 2 social network datasets. 
\begin{itemize}
	\item LastFM asia social network data set~\cite{Leskovec2014,Rozemberczki2020}: Vertices (the number of vertices $ n = 7624 $) represent LastFM users from asian countries and edges (the number of edges $ e = 27806 $) represent mutual follower relationships. We treat 18 different location of users, which are derived from the country field for each user, as the potential block. 
	\item Facebook large page-page network data set~\cite{Leskovec2014,Rozemberczki2019}: Vertices (the number of vertices $ n = 22470 $) represent official Facebook pages and edges (the number of edges $ e = 171002 $) represent mutual likes. We treat 4 page types \{Politician,~Governmental~Organization,~Television~Show,~Company\}, which are defined by Facebook, as the potential block.
\end{itemize}

We consider initial sampling parameter $ p_0 \in \left\{0.15, \; 0.35 \right\} $ and dynamic network sampling parameter $ p_1 \in \left\{0.05, \; 0.1, \; 0.15, \; 0.2, \; 0.25 \right\} $. For each $ p_1 $, we sample 100 times and compare the overall performance of Algorithm~\ref{algo:1} and~\ref{algo:2}. Figure~\ref{fig:real3} shows the results where ARI is reported as \texttt{mean($\pm$stderr)}. Again it suggests that given the same amout of resources, the proposed Chernoff-optimal dynamic network sampling scheme can yield better block recovery results. In other words, to reach the same level of performance, in terms of the empirical clustering results, the proposed Chernoff-optimal dynamic network sampling scheme needs less resources.

\begin{figure}[h!]
	\subfigure[LastFM: $ p_0 = 0.15, \; p_1 \in \left\{0.05, \; 0.1, \; 0.15, \; 0.2, \; 0.25 \right\}  $ \label{fig:real3a}]{
		\includegraphics[width=0.45\textwidth]{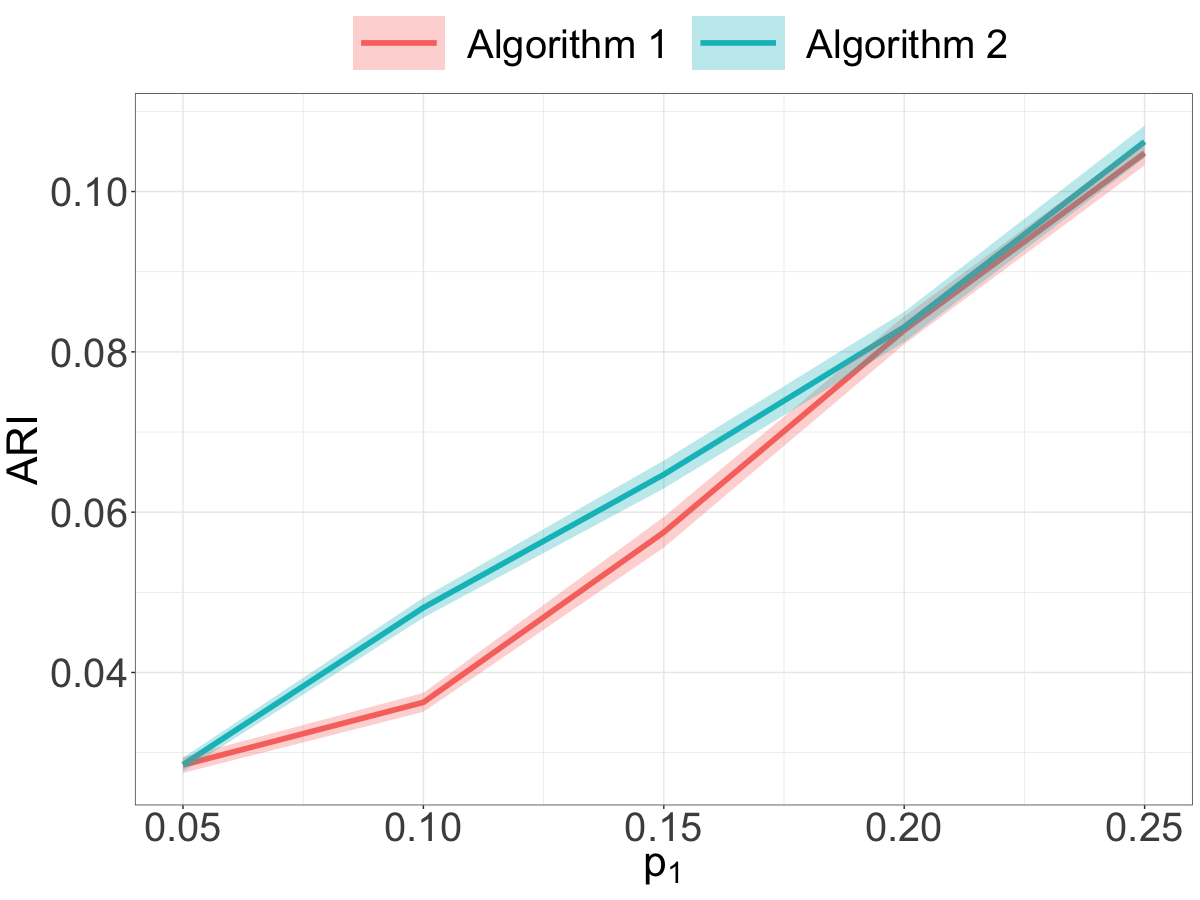}
	}
	\hfil
	\subfigure[Facebook: $ p_0 = 0.35, \; p_1 \in \left\{0.05, \; 0.1, \; 0.15, \; 0.2, \; 0.25 \right\}  $ \label{fig:real3b}]{
		\includegraphics[width=0.45\textwidth]{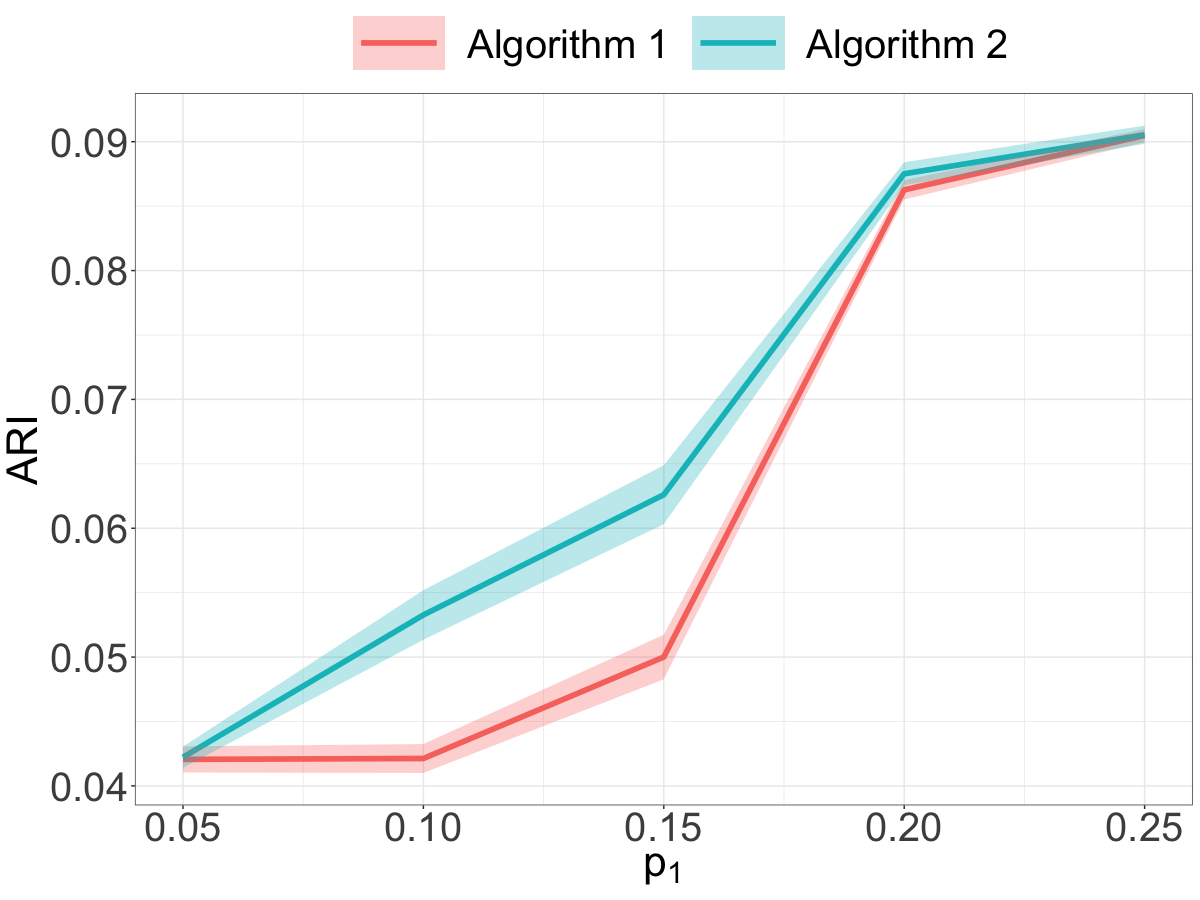}
	}
	\caption{Algorithms' comparative performance on social network data via ARI with different initial sampling parameter $ p_0 $ and dynamic network sampling parameter $ p_1 $.}
	\label{fig:real3}
\end{figure}

\section{Discussion}

\label{sec:6}

We propose a dynamic network sampling scheme to optimize block recovery for SBM when we only have a limited budget to observe a partial graph. Theoretically, we provide justification of our proposed Chernoff-optimal dynamic sampling scheme via the Chernoff information. Practically, we evaluate the performance, in terms of block recovery (community detection), of our method on several real datasets including diffusion MRI connectome dataset, Microsoft bing entity graph transitions dataset and social network datasets. Both theoretically and practically results suggest that our method can identify vertices that have the most impact on block structure and only check whether there are edges between them to save significant resources but still recover the block structure.

As the Chernoff-optimal dynamic sampling scheme depends on the initial clustering results to identify Chernoff-active blocks and construct dynamic edge set. Thus the performance could be impacted if the initial clustering is not very ideal. One of the future direction is to design certain strategy to reduce this dependency such that the proposed scheme is more robust.

\section*{Appendix}

\begin{proof}[Proof of Theorem~\ref{thm:Chernoff-Superiority}.]
	Let $ \mathbf{B} = \mathbf{U} \mathbf{S} \mathbf{U}^\top $ be the spectral decomposition of $ \mathbf{B} $ and $ \mathbf{B}^\prime = p \mathbf{B} $ with $ p \in (0, 1) $. Then we have
	\begin{equation}
	\label{eq:Bprime}
	\mathbf{B}^\prime = \mathbf{U}^\prime \mathbf{S} \left(\mathbf{U}^\prime\right)^\top \qquad \text{where} \qquad \mathbf{U}^\prime = \sqrt{p} \mathbf{U}.
	\end{equation}
	
	By Remark~\ref{remark:GRDPG-SBM}, to represent these two SBMs parametrized by two block connectivity matrices $ \mathbf{B} $ and $ \mathbf{B}^\prime $ respectively (with the same block assignment probability vector $ \boldsymbol{\pi} $) in the GRDPG models, we can take
	\begin{equation}
	\label{eq:nunuprime}
	\begin{split}
	\boldsymbol{\nu} & = \begin{bmatrix}
	\boldsymbol{\nu}_1 & \cdots & \boldsymbol{\nu}_K
	\end{bmatrix}^\top = \mathbf{U} |\mathbf{S}|^{1/2} \in \mathbb{R}^{K \times d}, \\
	\boldsymbol{\nu}^\prime & = \begin{bmatrix}
	\boldsymbol{\nu}_1^\prime & \cdots & \boldsymbol{\nu}_K^\prime
	\end{bmatrix}^\top = \mathbf{U}^\prime |\mathbf{S}|^{1/2} = \sqrt{p} \mathbf{U} |\mathbf{S}|^{1/2} = \sqrt{p} \boldsymbol{\nu} \in \mathbb{R}^{K \times d}.
	\end{split}
	\end{equation}
	
	Then for any $ k \in \{1, \cdots, K \} $, we have $ \boldsymbol{\nu}_k^\prime = \sqrt{p} \boldsymbol{\nu}_k \in \mathbb{R}^{d} $. By Theorem~\ref{thm:CLT-ASE-SBM}, we have
	\begin{equation}
	\begin{split}
	\boldsymbol{\Delta} & = \sum_{k=1}^{K} \pi_k \boldsymbol{\nu}_k \boldsymbol{\nu}_k^\top \in \mathbb{R}^{d \times d}, \\
	\boldsymbol{\Delta}^\prime & = \sum_{k=1}^{K} \pi_k \boldsymbol{\nu}_k^\prime \left(\boldsymbol{\nu}_k^\prime\right)^\top = p \sum_{k=1}^{K} \pi_k \boldsymbol{\nu}_k \boldsymbol{\nu}_k^\top = p \boldsymbol{\Delta} \in \mathbb{R}^{d \times d}.
	\end{split}
	\end{equation}
	
	Note that $ \mathbf{B} $ and $ \mathbf{B}^\prime $ have the same eigenvalues, thus we have $ \mathbf{I}_{d_+ d_-} = \mathbf{I}_{d_+ d_-}^\prime $. See also Lemma 2~\cite{Gallagher2019}. Then for $ k \in \{1, \cdots, K \} $, we have
	\begin{equation}
	\begin{split}
	\boldsymbol{\Sigma}_k & = \mathbf{I}_{d_+ d_-} \boldsymbol{\Delta}^{-1} \mathbb{E} \left[ \left(\boldsymbol{\nu}_k^\top \mathbf{I}_{d_+ d_-} \boldsymbol{\nu} \right) \left(1-\boldsymbol{\nu}_k^\top \mathbf{I}_{d_+ d_-} \boldsymbol{\nu} \right) \boldsymbol{\nu} \boldsymbol{\nu}^\top \right] \boldsymbol{\Delta}^{-1} \mathbf{I}_{d_+ d_-} \\
	& = \mathbf{I}_{d_+ d_-} \boldsymbol{\Delta}^{-1} \left[\sum_{\ell=1}^{K} \pi_{\ell} \left(\boldsymbol{\nu}_k^\top \mathbf{I}_{d_+ d_-} \boldsymbol{\nu}_{\ell} \right) \left(1-\boldsymbol{\nu}_k^\top \mathbf{I}_{d_+ d_-} \boldsymbol{\nu}_{\ell} \right) \boldsymbol{\nu}_{\ell} \boldsymbol{\nu}_{\ell}^\top \right] \boldsymbol{\Delta}^{-1} \mathbf{I}_{d_+ d_-} \in \mathbb{R}^{d \times d}, \\[1em]
	\boldsymbol{\Sigma}_k^{\prime} & = \frac{1}{p^2} \mathbf{I}_{d_+ d_-} \boldsymbol{\Delta}^{-1} \left[p^2 \sum_{\ell=1}^{K} \pi_{\ell} \left(\boldsymbol{\nu}_k^\top \mathbf{I}_{d_+ d_-} \boldsymbol{\nu}_{\ell} \right) \left(1-p \boldsymbol{\nu}_k^\top \mathbf{I}_{d_+ d_-} \boldsymbol{\nu}_{\ell} \right) \boldsymbol{\nu}_{\ell} \boldsymbol{\nu}_{\ell}^\top \right] \boldsymbol{\Delta}^{-1} \mathbf{I}_{d_+ d_-} \\
	& = \mathbf{I}_{d_+ d_-} \boldsymbol{\Delta}^{-1} \left[p \sum_{\ell=1}^{K} \pi_{\ell} \left(\boldsymbol{\nu}_k^\top \mathbf{I}_{d_+ d_-} \boldsymbol{\nu}_{\ell} \right) \left(1-\boldsymbol{\nu}_k^\top \mathbf{I}_{d_+ d_-} \boldsymbol{\nu}_{\ell} \right) \boldsymbol{\nu}_{\ell} \boldsymbol{\nu}_{\ell}^\top \right] \boldsymbol{\Delta}^{-1} \mathbf{I}_{d_+ d_-} \\
	& \qquad + \mathbf{I}_{d_+ d_-} \boldsymbol{\Delta}^{-1} \left[(1-p) \sum_{\ell=1}^{K} \pi_{\ell} \left(\boldsymbol{\nu}_k^\top \mathbf{I}_{d_+ d_-} \boldsymbol{\nu}_{\ell} \right) \boldsymbol{\nu}_{\ell} \boldsymbol{\nu}_{\ell}^\top \right] \boldsymbol{\Delta}^{-1} \mathbf{I}_{d_+ d_-} \\
	& = p \boldsymbol{\Sigma}_k + \mathbf{V}^\top \mathbf{D}_k(p) \mathbf{V} \in \mathbb{R}^{d \times d},
	\end{split}
	\end{equation}
	where
	\begin{equation}
	\begin{split}
	\mathbf{V} & = \boldsymbol{\nu} \boldsymbol{\Delta}^{-1} \mathbf{I}_{d_+ d_-} \in \mathbb{R}^{K \times d}, \\
	\mathbf{D}_k(p) & = (1-p) \text{diag} \left(\pi_1 \boldsymbol{\nu}_k^\top \mathbf{I}_{d_+ d_-} \boldsymbol{\nu}_1, \cdots, \pi_K \boldsymbol{\nu}_k^\top \mathbf{I}_{d_+ d_-} \boldsymbol{\nu}_K \right) \in (0, 1)^{K \times K}.
	\end{split}
	\end{equation}
	
	Recall that by Remark~\ref{remark:GRDPG-SBM}, we have $ \boldsymbol{\nu}_k^\top \mathbf{I}_{d_+ d_-} \boldsymbol{\nu}_\ell = \mathbf{B}_{k \ell} \in (0, 1) $ for all $ k, \ell \in \{ 1, \cdots, K \} $. Then we have $ \mathbf{D}_k(p) $ is positive-definite for any $ k \in \{1, \cdots, K \} $ and $ p \in (0, 1) $. For $ k, \ell \in \{1, \cdots, K \} $ and $ t \in (0, 1) $, let $ \boldsymbol{\Sigma}_{k\ell}(t) $ and $ \boldsymbol{\Sigma}_{k\ell}^{\prime}(t) $ denote the matrics as in Eq.~\eqref{eq:rhoapprox} corresponding to $ \mathbf{B} $ and $ \mathbf{B}^\prime $ respectively, i.e.,
	\begin{equation}
	\begin{split}
	\boldsymbol{\Sigma}_{k\ell}(t) & = t \boldsymbol{\Sigma}_k + (1-t) \boldsymbol{\Sigma}_\ell \in \mathbb{R}^{d \times d}, \\[1em]
	\boldsymbol{\Sigma}_{k\ell}^{\prime}(t) & = t \boldsymbol{\Sigma}_k^{\prime} + (1-t) \boldsymbol{\Sigma}_\ell^{\prime} \\
	& = t \left[p \boldsymbol{\Sigma}_k + \mathbf{V}^\top \mathbf{D}_k(p) \mathbf{V} \right] + (1-t) \left[p \boldsymbol{\Sigma}_\ell + \mathbf{V}^\top \mathbf{D}_\ell(p) \mathbf{V} \right] \\
	& = p \left[t \boldsymbol{\Sigma}_k + (1-t) \boldsymbol{\Sigma}_\ell \right] + \mathbf{V}^\top \left[t \mathbf{D}_k(p) + (1-t) \mathbf{D}_\ell(p) \right] \mathbf{V} \\
	& = p \boldsymbol{\Sigma}_{k\ell}(t) + \mathbf{V}^\top \mathbf{D}_{k \ell}(p, t) \mathbf{V} \in \mathbb{R}^{d \times d},
	\end{split}
	\end{equation}
	where
	\begin{equation}
	\mathbf{D}_{k \ell}(p, t) = t \mathbf{D}_k(p) + (1-t) \mathbf{D}_\ell(p) \in \mathbb{R}_+^{K \times K}.
	\end{equation}
	
	Recall that $ \mathbf{D}_k(p) $ and $ \mathbf{D}_\ell(p) $ are both positive-definite for any $ k, \ell \in \{1, \cdots, K \} $ and $ p \in (0, 1) $, thus $ \mathbf{D}_{k \ell}(p, t) $ is also positive-definite for any $ k, \ell \in \{1, \cdots, K \} $ and $ p, t \in (0, 1) $. Now by the Sherman-Morrison-Woodbury formula~\cite{Horn2012}, we have
	\begin{equation}
	\begin{split}
	\left[\boldsymbol{\Sigma}_{k\ell}^{\prime}(t) \right]^{-1} & = \left[p \boldsymbol{\Sigma}_{k\ell}(t) + \mathbf{V}^\top \mathbf{D}_{k \ell}(p, t) \mathbf{V} \right]^{-1} \\
	& = \frac{1}{p} \boldsymbol{\Sigma}_{k\ell}^{-1}(t) - \frac{1}{p^2} \boldsymbol{\Sigma}_{k\ell}^{-1}(t) \mathbf{V}^\top \left[\mathbf{D}_{k \ell}^{-1}(p, t) + \frac{1}{p} \mathbf{V} \boldsymbol{\Sigma}_{k\ell}^{-1}(t) \mathbf{V}^\top \right]^{-1} \mathbf{V} \boldsymbol{\Sigma}_{k\ell}^{-1}(t) \\
	& = \frac{1}{p} \boldsymbol{\Sigma}_{k\ell}^{-1}(t) - \frac{1}{p^2} \boldsymbol{\Sigma}_{k\ell}^{-1}(t) \mathbf{V}^\top \mathbf{M}_{k \ell}^{-1}(p, t)\mathbf{V} \boldsymbol{\Sigma}_{k\ell}^{-1}(t) \in \mathbb{R}^{d \times d},
	\end{split}
	\end{equation}
	where
	\begin{equation}
	\mathbf{M}_{k \ell}(p, t) = \mathbf{D}_{k \ell}^{-1}(p, t) + \frac{1}{p} \mathbf{V} \boldsymbol{\Sigma}_{k\ell}^{-1}(t) \mathbf{V}^\top  \in \mathbb{R}^{K \times K}.
	\end{equation}
	
	Recall that for any $ k, \ell \in \{1, \cdots, K \} $ and $ p, t \in (0, 1) $, $ \mathbf{D}_{k \ell}(p, t) $ and $ \boldsymbol{\Sigma}_{k\ell}(t) $ are both positive-definite, thus $ \mathbf{M}_{k \ell}(p, t) $ is also positive-definite. Then for any $ k, \ell \in \{1, \cdots, K \} $ and $ p,t \in (0, 1) $, we have
	\begin{equation}
	\label{eq:nuSigma}
	\begin{split}
	(\boldsymbol{\nu}_k^\prime - \boldsymbol{\nu}_\ell^\prime)^\top \left[\boldsymbol{\Sigma}_{k\ell}^{\prime}(t) \right]^{-1} (\boldsymbol{\nu}_k^\prime - \boldsymbol{\nu}_\ell^\prime) & = p (\boldsymbol{\nu}_k - \boldsymbol{\nu}_\ell)^\top \\
	& \quad \left[\frac{1}{p} \boldsymbol{\Sigma}_{k\ell}^{-1}(t) - \frac{1}{p^2} \boldsymbol{\Sigma}_{k\ell}^{-1}(t) \mathbf{V}^\top \mathbf{M}_{k \ell}^{-1}(p, t) \mathbf{V} \boldsymbol{\Sigma}_{k\ell}^{-1}(t) \right] \\
	& \quad (\boldsymbol{\nu}_k - \boldsymbol{\nu}_\ell) \\
	& = (\boldsymbol{\nu}_k - \boldsymbol{\nu}_\ell)^\top \boldsymbol{\Sigma}_{k\ell}^{-1}(t) (\boldsymbol{\nu}_k - \boldsymbol{\nu}_\ell) \\
	& \quad - \frac{1}{p} \mathbf{x}^\top \mathbf{M}_{k \ell}^{-1}(p, t) \mathbf{x} \\
	& = (\boldsymbol{\nu}_k - \boldsymbol{\nu}_\ell)^\top \boldsymbol{\Sigma}_{k\ell}^{-1}(t) (\boldsymbol{\nu}_k - \boldsymbol{\nu}_\ell) - h_{k \ell}(p, t),
	\end{split}
	\end{equation}
	where
	\begin{equation}
	\begin{split}
	\mathbf{x} & = \mathbf{V} \boldsymbol{\Sigma}_{k\ell}^{-1}(t) (\boldsymbol{\nu}_k - \boldsymbol{\nu}_\ell) \in \mathbb{R}^K, \\
	h_{k \ell}(p, t) & = \frac{1}{p} \mathbf{x}^\top \mathbf{M}_{k \ell}^{-1}(p, t) \mathbf{x}.
	\end{split}
	\end{equation}
	
	Recall that for any $ k, \ell \in \{1, \cdots, K \} $ and $ p, t \in (0, 1) $, $ \mathbf{M}_{k \ell}(p, t) $ is positive-definite, thus we have $ h_{k \ell}(p, t) > 0 $. Together with Eq.~\eqref{eq:nuSigma}, we have
	\begin{equation}
	t (1-t) (\boldsymbol{\nu}_k - \boldsymbol{\nu}_\ell)^\top \boldsymbol{\Sigma}_{k\ell}^{-1}(t) (\boldsymbol{\nu}_k - \boldsymbol{\nu}_\ell) > t (1-t) (\boldsymbol{\nu}_k^\prime - \boldsymbol{\nu}_\ell^\prime)^\top \left[\boldsymbol{\Sigma}_{k\ell}^{\prime}(t) \right]^{-1} (\boldsymbol{\nu}_k^\prime - \boldsymbol{\nu}_\ell^\prime).
	\end{equation}
	
	Thus for any $ k, \ell \in \{1, \cdots, K \} $, we have
	\begin{equation}
	\begin{split}
	C_{k ,\ell}(\mathbf{B}, \boldsymbol{\pi}) & =\sup_{t \in (0, 1)} \left[ t (1-t) (\boldsymbol{\nu}_k - \boldsymbol{\nu}_\ell)^\top \boldsymbol{\Sigma}_{k\ell}^{-1}(t) (\boldsymbol{\nu}_k - \boldsymbol{\nu}_\ell) \right], \\
	& > \sup_{t \in (0, 1)} \left[ t (1-t) (\boldsymbol{\nu}_k^\prime - \boldsymbol{\nu}_\ell^\prime)^\top \left[\boldsymbol{\Sigma}_{k\ell}^{\prime}(t) \right]^{-1} (\boldsymbol{\nu}_k^\prime - \boldsymbol{\nu}_\ell^\prime) \right] \\
	& = C_{k ,\ell}(\mathbf{B}^\prime, \boldsymbol{\pi}).
	\end{split}
	\end{equation}
	
	Let $ \rho_B $ and $ \rho_{B^\prime} $ denote the Chernoff information obtained as in Eq.~\eqref{eq:rhoapprox} corresponding to $ \mathbf{B} $ and $ \mathbf{B}^\prime $ respectively (with the same block assignment probability vector $ \boldsymbol{\pi} $). Then we have
	\begin{equation}
	\rho_{B} \approx \min_{k \neq l} C_{k ,\ell}(\mathbf{B}, \boldsymbol{\pi}) > \min_{k \neq l} C_{k ,\ell}(\mathbf{B}^\prime, \boldsymbol{\pi}) \approx \rho_{B^\prime}. 
	\end{equation}
	
	Thus we have $ \mathbf{B} \succ \mathbf{B}^\prime = p \mathbf{B} $ for $ p \in (0, 1) $.
\end{proof}

\begin{proof}[Proof of Corollary~\ref{cor:Chernoff-Superiority}.]
	By Eq.~\eqref{eq:B0} and Eq.~\eqref{eq:B1}, we have
	\begin{equation}
	\begin{split}
	\mathbf{B}_0 & = \frac{p_0}{p_0+p_1} \mathbf{B}_1, \\
	\mathbf{B}_1 & = (p_0 + p_1) \mathbf{B}.
	\end{split}
	\end{equation}
	
	Recall that $ p_0 \in (0, 1) $ and $ p_1 \in (0, 1-p_0) $. Then by Theorem~\ref{thm:Chernoff-Superiority}, we have $ \mathbf{B} \succ \mathbf{B}_1 \succ \mathbf{B}_0 $.
\end{proof}


\begin{backmatter}

\section*{Acknowledgements}
 This problem was posed to us by Adam Cardinal-Stakenas and Kevin Hoover.

\section*{Funding}
Cong Mu's work is partially supported by the Johns Hopkins Mathematical Institute for Data Science (MINDS) Data Science Fellowship.

\section*{Abbreviations}
\textbf{SBM}: Stochastic Blockmodel \\
\textbf{GRDPG}: Generalized Random Dot Product Graph \\
\textbf{ASE}: Adjacency Spectral Embedding \\
\textbf{LSE}: Laplacian Spectral Embedding \\
\textbf{GMM}: Gaussian Mixture Modeling \\
\textbf{BIC}: Bayesian Information Criterion \\
\textbf{ARI}: Adjusted Rand Index \\
\textbf{stderr}: Standard Error \\
\textbf{NDMG}: NeuroData’s Magnetic Resonance Imaging to Graphs

\section*{Availability of data and materials}
Social network datasets are available at \texttt{http://snap.stanford.edu/data/}.


\section*{Competing interests}
The authors declare that they have no competing interests.


\section*{Authors' contributions}
Cong Mu developed the theory, designed \& implemented the methods, conducted the experiments, and wrote the manuscript . Youngser Park implemented the methods, conducted the experiments, and edited the manuscript. Carey E. Priebe formulated the problem, designed the methods, developed the theory and edited the manuscript. All authors read and approved the manuscript.

\section*{Authors' information}


\bibliographystyle{bmc-mathphys} 
\bibliography{bmc_article}      







\end{backmatter}
\end{document}